\documentclass{amsart}

\usepackage{amsfonts,amsmath,graphicx,subfigure,multirow}
\usepackage[margin=3.5cm]{geometry}

\newtheorem{theorem}{Theorem}
\theoremstyle{plain}

\newtheorem{corollary}{Corollary}

\newtheorem{lemma}{Lemma}

\newtheorem{proposition}{Proposition}
\newtheorem{remark}{Remark}

\numberwithin{equation}{section}

\begin{document}

\title[Pattern Formation and Dynamic Transition for MHD]{Pattern Formation and Dynamic Transition for Magnetohydrodynamic Convection}

\author[Sengul]{Taylan Sengul}
\address[TS]{Department of Mathematics,
Indiana University, Bloomington, IN 47405}
\email{msengul@indiana.edu}

\author[Wang]{Shouhong Wang}
\address[SW]{Department of Mathematics,
Indiana University, Bloomington, IN 47405}
\email{showang@indiana.edu, http://www.indiana.edu/~fluid}

\thanks{The work  was supported in part by the
Office of Naval Research and by the National Science Foundation.}


\date{\today}

\begin{abstract}
The main objective of this paper is to describe the dynamic transition of the incompressible MHD equations in a three dimensional (3D) rectangular domain from a perspective of pattern formation. First it is shown that the system always undergoes a dynamical transition as the Rayleigh number crosses the first critical threshold which corresponds to the first eigenvalues of the linearized problem around the basic state. The type of transitions is determined by the nonlinear interactions.

When the magnetic Prandtl number $\mathfrak{p}_2\geq1$ or the Chandrasekhar number $Q<Q_0$ where $Q_0$ depends on the system  parameters, the first eigenvalues are always real. Generically, the first eigenvalue is simple and a Type-I or a Type-II transition is possible depending on a non-dimensional number exactly given in terms of the system parameters. In particular, the rolls are stable patterns after the transition if the Chandrasekhar number $Q>Q_{\ast}$ or the magnetic Prandtl number $\mathfrak{p}_2>\mathfrak{p}_{\ast}$ where $Q_{\ast}$ and $\mathfrak{p}_{\ast}$ depends on the length scales of the domain. When the horizontal length scales $L_1$ and $L_2$ satisfy $\sqrt{3}kL_1=jL_2$ for some positive integers $j$ and $k$, two modes, characterizing a hexagonal pattern, become unstable simultaneously. In this case, we show that all three types of transition are possible. Depending on two non-dimensional parameters, we show that there are eight different transition scenarios. In this case we also show that for $\mathfrak{p}_2\geq 8$, the transition is always Type-I with rolls and rectangles as stable patterns and hexagons as unstable patterns after the transition. However $\mathfrak{p}_2\geq 8$ is a crude estimate and our numerical investigation suggests that this type of Type-I transition will be preferred for $\mathfrak{p}_2\geq\mathfrak{p}_2^{\ast}$ where $\mathfrak{p}_2^{\ast}<2.24$.

In the case where  $\mathfrak{p}_2<1$ and $Q>Q_0$, the first eigenvalues are complex, and the system undergoes a dynamical transition to spatiotemporal oscillatory states. In particular, when the first critical modes have time periodic roll structure, we show that the system undergoes a Type-I or Type-II transition. In the large Chandrasekhar number limit or the small oscillation frequency limit, the transition is always Type-I and the time periodic rolls are stable after the transition.
\end{abstract}

\maketitle

\section{Introduction}
The Rayleigh--B\'{e}nard convection is a fundamental problem of natural
convective heat transfer which is characterized by a vertical temperature
gradient aligned with the acceleration of gravity being maintained over a
horizontal layer of fluid. Due to the thermal expansion, the fluid is
heavier at the top and lighter at the bottom. As the temperature difference
between the lower and upper fluid boundary exceeds a critical level, a
convective motion sets in. There have been numerous studies concerning the stability and instability of the convective flows.  A detailed discussion can be found in \cite{chandrasekhar,cross93,getling1998rayleigh, kosch}.

External magnetic fields change the characteristics of this convection
significantly for electrically well conducting fluids. First it is well
known that the critical Rayleigh number and the wave number increase with
an increasing Chandrasekhar number $Q$ for the onset of convection. Physically
this is due to the fact that the energy released by the buoyancy force
acting on the fluid must balance the energy dissipated by not only the
viscosity but also the Joule heating. Thus the magnetic field imposes stability to the fluid. Second, the existence of a magnetic field
allows both the steady and oscillatory convections; see \cite{chandrasekhar,
Proctor1982}.

We aim to describe the dynamic stability and transition of the magnetic
convection for an incompressible fluid in a rectangular domain in $\mathbb{R}^{3}$ from a pattern formation perspective. As is well known, for the magnetohydrodynamics (MHD) equations, due to non-selfadjoint linear operator,  the transition can be caused by a finite set of real or complex eigenvalues crossing the imaginary axis. This makes transitions to both spatially periodic time independent states and to spatiotemporally periodic states possible. We focus on the formation of patterns having roll, rectangular and hexagonal structures. Our main goal is to precisely determine the type of transitions associated with these patterns and hence characterize the stability of these patterns in terms of the parameters of the system.

When the first eigenvalue is real and simple, the transition can only be Type-I or Type-II depending on a number exactly given in terms of the system parameters. In particular, when the first critical eigenmode has a roll structure, the type of transition is independent of the Prandtl number $\mathfrak{p}_1$. The transition, in this case, is always Type-I if the Chandrasekhar number $Q>Q_{\ast}$ or the magnetic Prandtl number $\mathfrak{p}_2>\mathfrak{p}_{\ast}$ where $Q_{\ast}$ and $\mathfrak{p}_{\ast}$ depends on the length scales of the domain. We find that $Q_{\ast}<307$ and $\mathfrak{p}_{\ast}<2.24$ regardless of the length scales of the domain and  $Q_{\ast}\rightarrow4\pi^2$ and $\mathfrak{p}_{\ast}\rightarrow2/\sqrt{3}$ as $\max\{L_1,L_2\}\rightarrow\infty$ where $L_1$ and $L_2$ are the horizontal length scales.

Next we study the case where there are two critical real eigenvalues. In this case we only consider the special geometry 
\[
\frac{L_1}{L_2}=\frac{j}{k \sqrt{3}}. \] with positive integers $j$, $k$ and $L_1$, $L_2$ denoting the horizontal length scales of the box.
With this assumption, it is possible that  two modes which can characterize a hexagon pattern become unstable at the same critical parameter. In this case we find that all types of transitions are possible in a total of eight different transition scenarios. However, in our numerical investigation, we encountered only two of these scenarios. Fixing  $Q<Q_{\ast}$, we found that the system moves from a Type-III transition regime to a Type-I regime as $\mathfrak{p}_2$ crosses $\mathfrak{p}_{\ast}$.  In the Type-I transition regime, the transition state is an attractor homeomorphic to a circle containing eight steady states and the connecting heteroclinic orbits. Of these eight steady states, those having roll and rectangular patterns are stable whereas the ones with hexagonal pattern are saddles. In the Type-III transition regime, a neighborhood of the basic solution is split into two sectorial regions $U_1$ and $U_2$. In $U_1$, the initial conditions move  out of this neighborhood whereas in $U_2$, the initial conditions will move to an attractor which contains a rectangular pattern as a minimal attractor. In this case, we can single out the transition scenario if $\mathfrak{p}_2\geq 8$. We prove that when $\mathfrak{p}_2\geq8$, the transition is always Type-I with rolls and rectangles as stable states while hexagons as unstable states after the transition. However $\mathfrak{p}_2\geq 8$ is a crude estimate and our numerical investigation suggests that this type of Type-I transition will be preferred for $\mathfrak{p}_2\geq\mathfrak{p}_2^{\ast}$ where $\mathfrak{p}_2^{\ast}<2.24$. It is worth mentioning that when the length scales are in close vicinity of the above relation, generically only one of the eigenvalues that characterize a hexagonal structure will be unstable. However a second negative eigenvalue will be very close to zero at the critical parameter and the structure of transition described above case will still persist.

Finally, we consider the case where a pair of complex eigenvalues become unstable simultaneously. In this case, we only consider a roll type critical eigenmode. We show that the first transition can be Type-I or Type-II. In particular for $Q$ sufficiently large or for the oscillation frequency $\rho$ sufficiently small the transition is Type-I and the transition structure is a time periodic roll pattern.

There have been extensive studies on the MHD convection problem, see among others  \cite{chandrasekhar, Proctor1982}  and the references therein. There are several features of this work that distinguishes it from the other studies.

First, to our knowledge none of the previous studies deals with nonlinear  3D equations in a finite box. As we will discuss, the presence of side walls brings in additional difficulties and interesting dynamical features. From a physical point of view the addition of side walls is crucial especially when the horizontal scales are  comparable to the depth of the layer.

Second, the problem is studied from a transition point of view instead of bifurcation point of view. The key philosophy of the dynamic transition theory which is recently
developed by Ma \& Wang \cite{b-book,ptd} is to search for the full set of transition
states, giving a complete characterization on stability and transition. The
set of transition states is represented by a local attractor rather than
some steady states or periodic solutions or other types of orbits as part of
this local attractor.

Third, a crucial step in describing the transition of the system is the reduction of the problem on the center manifold. In this paper we use a slightly different method for the reduction of the problem which is worth mentioning here. Usually, one expands the center manifold in terms of the eigenfunctions of the linear operator of the system. This allows the utilization of the approximations for the center manifold in a natural way. However, as in the case of the magnetohydrodynamics convection, this approximation can be difficult to explicitly calculate when the eigenfunctions of the linear operator of the original system is complicated for computational purposes. We overcome this difficulty by defining a new set of basis for the functional space of the problem which behaves nicely under the original linear operator.

The paper is organized as follows, in Section 2, the mathematical setting of the problem is introduced. In  Section 3, linear theory is summarized. In Section 4 we state the main theorems. We give the proofs of the main theorems in Section 5.  In Section 6, physical remarks and conclusions are discussed.  
\section{Mathematical Setting}

We consider thermally driven convection of an electrically conducting fluid
in the presence of a magnetic field in a rectangular domain $\Omega
=(0,l_{1})\times (0,l_{2})\times (0,h)$ in $\mathbb{R}^{3}$. Subject to
Boussinesq approximation (see \cite{chandrasekhar} among others), the
evolution equations read:
\begin{equation}
\label{NondimEqu1}
\begin{aligned} 
& \frac{\partial u}{\partial t}+(u\cdot \nabla)u =
-\frac{1}{\rho _{0}} (\nabla p+\rho g\vec{k})+
\mu_{0}\left( \nabla \times H\right) \times H+\nu \Delta u,\\ 
& \frac{\partial H}{\partial t}+\left( u\cdot \nabla \right) H = 
\left( H\cdot \nabla \right)u+\eta \Delta H,\\ 
& \frac{\partial T}{\partial t}+\left( u\cdot \nabla\right) T =
\kappa \Delta T,\\ 
& \text{div} u = \text{div} H=0. 
\end{aligned}
\end{equation}
Here $u=\left( u_{1},u_{2},u_{3}\right) $ is the velocity field, $H$ is the
magnetic field, $T$ is the temperature, $\vec{k}=\left( 0,0,1\right) $, $g$
is the gravitational acceleration, $\nu $ is the kinematic viscosity, $\mu
_{0}$ is the magnetic permeability, $\eta $ is the magnetic diffusivity
(also called resistivity), $\kappa $ is the thermal diffusivity, $T_{0}$ is
the reference temperature at $x_{3}=0$, $\rho _{0}$ is the density at $T_{0}$,
and $a>0$ is the coefficient of thermal expansion. The fluid density $\rho$ is 
given by the equation of state:
\begin{equation}
\rho=\rho_{0}[1-a(T-T_{0})].
\label{equofstate}
\end{equation}

Although the case where the imposed magnetic field and the gravitational
field act in different directions is also interesting, for simplicity, we
will assume that they are parallel. Denoting the temperature at $x_{3}=h$ by 
$T_{1}$, the basic state of (\ref{NondimEqu1}) is a motionless state given by:
\begin{align*}
\overline{u} &=0, \\
\overline{T} &=T_{0}+\left( T_{1}-T_{0}\right) \frac{x_{3}}{h}, \\
\overline{H} &=H_{0}\overrightarrow{k}, \\
\overline{p} &=p_{0}-\rho _{0}g\left( x_{3}+a\left( T_{0}-T_{1}\right) \frac{%
x_{3}^{2}}{2h}\right) .
\end{align*}%

To put the equations into non-dimensional form, we consider the deviation
of the solution from the basic state:%
\begin{equation*}
u^{\prime \prime }=u-\overline{u},\qquad T^{\prime \prime }=T-\overline{T}%
,\qquad H^{\prime \prime }=H-\overline{H},\qquad p^{\prime \prime }=p-%
\overline{p},
\end{equation*}
and we set
\begin{align*}
&u^{\prime \prime }=\frac{\kappa }{h}u^{\prime },\qquad 
& &H^{\prime \prime}=\frac{\kappa }{\eta }H_{0}H^{\prime }, \qquad 
& &T^{\prime \prime}=\left(T_{0}-T_{1}\right) T^{\prime }, \\
& x=hx^{\prime }, 
& &p^{\prime \prime }=\frac{\rho _{0}\nu \kappa }{h^{2}}%
p^{\prime }, & &t=\frac{h^{2}}{\kappa}t^{\prime }.
\end{align*}
Also we define the following nondimensional numbers:
\begin{equation}  
\label{Nondimensional-numbers}
\begin{aligned} 
& \mathfrak{p}_1 =\nu /\kappa & \text{\qquad the Prandtl number,} \\
& \mathfrak{p}_2 =\eta /\kappa & \text{ \qquad the magnetic Prandtl number,} \\ 
& R =\frac{ga\left( T_{0}-T_{1}\right) }{\kappa \nu }h^{3} &\text{ \qquad the Rayleigh number,} \\ 
& Q =\mu _{0}\frac{H_{0}^{2}h^{2}}{\nu \eta } & \text{ \qquad the Chandrasekhar number.} 
\end{aligned}
\end{equation}
Using the identity
\begin{equation*}
\left( \nabla \times H\right) \times H=-\frac{1}{2}\nabla \left\vert
H\right\vert ^{2}+\left( H\cdot \nabla \right) H,
\end{equation*}%
omitting the primes and denoting all the terms that can be written as
gradients by $p$, the equations (\ref{NondimEqu1})-(\ref{equofstate}) take the form:
\begin{equation}  \label{Eqn1}
\begin{aligned} 
&\frac{\partial u}{\partial t}+(u\cdot \nabla )u =
\mathfrak{p}_1\left(-\nabla p+RT \overrightarrow{k}+\Delta u+Q\frac{\partial H}{\partial x_{3}}+\frac{Q}{ \mathfrak{p}_2}\left( H\cdot \nabla \right) H\right) , \\
&\frac{\partial H}{\partial t}+\left( u\cdot \nabla \right) H =
\left( H\cdot \nabla \right) u+\mathfrak{p}_2\left( \frac{\partial u}{\partial x_{3}}+\Delta H\right) , \\
&\frac{\partial T}{\partial t}+\left( u\cdot \nabla \right) T = \Delta T+u_{3}, \\ 
& \text{div} u = \text{div} H=0. 
\end{aligned}
\end{equation}

The non-dimensional domain is $\Omega =\left( 0,L_{1}\right) \times \left(
0,L_{2}\right) \times \left( 0,1\right) $\ where $L_{1}=l_{1}/h$, $%
L_{2}=l_{2}/h$. We use the following idealized boundary conditions which are the free-slip boundary conditions for the velocity together with the condition that $H$
remains vertical at $x_{3}=0,1$.
\begin{equation}  
\label{B.C.}
\begin{aligned} 
& u_{1} =\frac{\partial u_{2}}{\partial x_{1}}=\frac{\partial u_{3}}{\partial x_{1}}= \frac{\partial T}{\partial x_{1}}=
H_{1}=\frac{\partial H_{2}}{\partial x_{1}}=\frac{\partial H_{3}}{\partial x_{1}}=0  & \text{at } &x_{1}=0,L_{1}, \\
& \frac{\partial u_{1}}{\partial x_{2}} =u_{2}=\frac{\partial u_{3}}{\partial x_{2}}=\frac{\partial T}{\partial x_{2}}=
\frac{\partial H_{1}}{\partial x_{2}}=H_{2}=\frac{\partial H_{3}}{\partial x_{2}}=0  & \text{at } &x_{2}=0,L_{2}, \\
&\frac{\partial u_{1}}{\partial x_{3}} =\frac{\partial u_{2}}{\partial x_{3}}=
u_{3}=T=H_{1}=H_{2}=\frac{\partial H_{3}}{\partial x_{3}}=0  & \text{at } &x_{3}=0,1. 
\end{aligned}
\end{equation}

We recall here the functional setting of (\ref{Eqn1})-(\ref{B.C.}) and refer
the interested readers to \cite{temam88}.
\begin{align*}
&H =\left\{ \psi=( u,H,T) \in L^{2}\left( \Omega \right) ^{7}\mid \text{div}u=\text{div}H=
u\cdot n\mid _{\partial \Omega }=H\cdot n\mid _{\partial \Omega }=0\right\} , \\
&H_{1} =\left\{ \psi=( u,H,T) \in H^{2}\left( \Omega \right) ^{7}\cap H\mid \psi 
\text{ satisfies (\ref{B.C.})}\right\} .
\end{align*}%
Let $L_{R}=A+B_{R}:H_{1}\rightarrow H$ and $G:H_{1}\rightarrow H$ be defined by
\begin{equation}\label{operators}
\begin{aligned}
&A\psi =\mathcal{P}\left( \mathfrak{p}_1\Delta u+\mathfrak{p}_1Q\frac{\partial H}{\partial x_{3}},%
\mathfrak{p}_2\frac{\partial u}{\partial x_{3}}+\mathfrak{p}_2\Delta H,\Delta T\right) , \\
&B_{R}\psi =\mathcal{P}\left( \mathfrak{p}_1RT\overrightarrow{k},0,u_{3}\right) , \\
&G\left( \psi \right) =\mathcal{P}\left( \frac{Q\mathfrak{p}_1}{\mathfrak{p}_2}\left(
H\cdot \nabla \right) H-(u\cdot \nabla )u,\left( H\cdot \nabla \right)
u-\left( u\cdot \nabla \right) H,-\left( u\cdot \nabla \right) T\right) .
\end{aligned}
\end{equation}
Here $\mathcal{P}:L^{7}\left( \Omega \right) \rightarrow H$ is the Leray
projector. For simplicity, we also use $G$ to represent the following trilinear form for $\psi _{i}=\left( u_{i},T_{i},H_{i}\right) \in H$, 
\begin{equation}\label{trilinear}
\begin{aligned}
G\left( \psi _{1},\psi _{2} ,\psi _{3}\right) =&\int_{\Omega}\left[
 \frac{Q\mathfrak{p}_1}{\mathfrak{p}_2}(H_1\cdot \nabla ) H_2-(u_1\cdot \nabla )u_2\right]\cdot u_3 \\
& +\int_{\Omega}\bigg[( H_1\cdot \nabla)
u_2-\left( u_1\cdot \nabla \right) H_2 \bigg]\cdot H_3-\left( u_1\cdot \nabla \right) T_2\, T_3.
\end{aligned}
\end{equation}
We also define:
\begin{equation} \label{Gs}
 G_s(\psi_1,\psi_2,\psi_3)= G(\psi_1,\psi_2,\psi_3)+ G(\psi_2,\psi_1,\psi_3).
\end{equation}
Now (\ref{Eqn1})-(\ref{B.C.}) can be written in the following functional form:
\begin{equation}
\begin{aligned}
&\frac{d\psi }{dt} =L_{R}\psi +G\left( \psi \right) , \\
&\psi \left( 0\right) =\psi _{0}.
\end{aligned}
\end{equation}

\section{Principle of Exchange of Stabilities}

The linear stability is determined by the critical crossing of the first
eigenvalues from the imaginary axis which is called the principle of
exchange of stabilities (PES). The linear theory associated with (\ref{Eqn1}%
)-(\ref{B.C.}) is well known and can be found for example in \cite%
{chandrasekhar}. Here we will briefly recall the necessary results.

First we consider the following eigenvalue problem:
\begin{equation}  \label{Lin-EV-Cpct}
L_R \psi=\beta(R)\psi,
\end{equation}
where $L_R$ is as defined in \eqref{operators}. Let us denote the set of admissible indices by
\[Z=\left\{ \left( j_1,j_2,j_3\right) \in \mathbb{Z}%
^{3}:j,k\geq 0\text{, }j_1^{2}+j_2^{2}\neq 0\text{, }j_3\geq 1\right\}. \notag
\]
It is well-known that the eigenvalues of \eqref{Lin-EV-Cpct} having indices $J=(j_1,j_2,j_3)\in Z$ satisfy the following polynomial equation:
\begin{align}
& \beta ^{3}+b_{2}^{J}\beta ^{2}+b_{1}^{J}\left( R\right) \beta
+b_{0}^{J}\left( R\right) =0,   \label{polynom} \\
\intertext{where}
& b_{2}^{J} =\left( \mathfrak{p}_1+\mathfrak{p}_2+1\right) \gamma _{J}^{2}, \notag \\
& b_{1}^{J} =\mathfrak{p}_1[( \mathfrak{p}_2+1+\frac{\mathfrak{p}_2}{\mathfrak{p}_1})
\gamma _{J}^{4}+\mathfrak{p}_2Q\left( j_3\pi \right) ^{2}-R\frac{\alpha
_{J}^{2}}{\gamma _{J}^{2}}] , \notag \\
& b_{0}^{J} =\mathfrak{p}_1\mathfrak{p}_2( \gamma _{J}^{6}+Q\left( j_3\pi )
^{2}\gamma _{J}^{2}-\alpha _{J}^{2}R\right) \notag .
\end{align}
Here 
\begin{equation}\label{wavenumbers}
\begin{aligned}
& \alpha _{J}^{2}=( j_1^{2}L_{1}^{-2}+j_2^{2}L_{2}^{-2}) \pi ^{2}, \\
& \gamma _{J}^{2}=\alpha _{J}^{2}+j_3^2\pi^2.
\end{aligned}
\end{equation}
As the linear operator \eqref{Lin-EV-Cpct} is not self-adjoint, the onset of the instability may correspond to either real or complex eigenvalues $\beta$, leading to transitions to states with either time independent structures or time periodic structures. We will denote the critical Rayleigh number by $R_r$ when the first eigenvalues are real, and by $R_c$ when the first eigenvalues are complex. Solving \eqref{polynom} for $R$ when $\beta=0$ and $Re\beta=0$ and minimizing over all admissible indices, we can find the expressions for $R_r$ and $R_c$ as follows:
\begin{align}
& R_{r} =\underset{J \in Z}{\min} 
\frac{\gamma  _{J}^{2}}{\alpha _{J}^{2}}
\left( \gamma_{J}^{4}+Q\left( j_3\pi \right) ^{2}\right),  \label{Rr} \\
\begin{split} \label{Rc}
& R_{c} =\underset{J \in Z}{\min} 
\frac{\left( \mathfrak{p}_2+1\right) \left( \mathfrak{p}_1+\mathfrak{p}_2\right)}{\mathfrak{p}_1} 
\frac{\gamma _{J}^{2}}{\alpha _{J}^{2}}
[ \gamma_{J}^{4}+\frac{\mathfrak{p}_2\mathfrak{p}_1}{\left( \mathfrak{p}_2+1\right) \left(\mathfrak{p}_1+1\right) }Q\left( j_3\pi \right) ^{2}]. 
\end{split}
\end{align}

The difference between an unbounded domain and a bounded one is that in the unbounded case the minimums in \eqref{Rr} and \eqref{Rc} are taken over all $\alpha_J>0$ whereas in the bounded case $\alpha_J$ has the specific form given above.
 Figure \ref{Figureminimizers} shows the minimizing indices of $R_r$ as a function of the length scales $L_1$ and $L_2$.
\begin{figure}
\subfigure[$Q=0$]{
\includegraphics[scale=0.5]{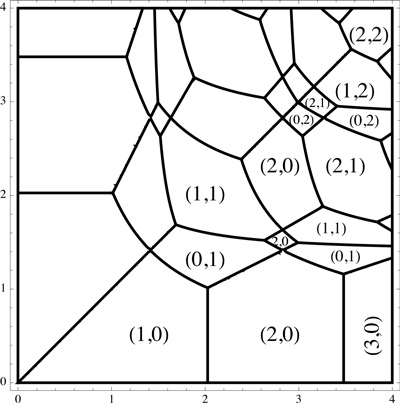}
}
\subfigure[$Q=10$]{
\includegraphics[scale=.4]{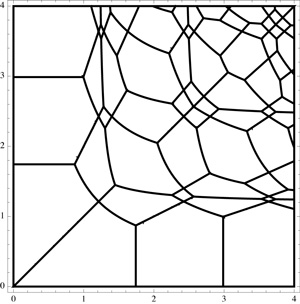}
}
\caption{Axes show the length scales $L_1$ and $L_2$. The numbers indicate the indices $(j_r,k_r)$ which minimize \eqref{Rr}.} 
\label{Figureminimizers}
\end{figure}

We note that for each $J\in Z$ there are three eigenvalues of \eqref{polynom} which we can order as $\beta^1_J\geq \beta^2_J\geq\beta^3_J$. The next proposition states that the principle of exchange of stabilities (PES) condition is valid.
\begin{proposition}\cite{chandrasekhar}\label{PrExSt}
For $\mathfrak{p}_2<1$, there exists $0<Q_{0}<\infty $ depending on $\mathfrak{p}_1$ and $%
\mathfrak{p}_2$ such that
\begin{equation}\label{Q0}
Q_0<\pi^2\frac{\mathfrak{p}_2(\mathfrak{p}_1+1)}{\mathfrak{p}_1(1-\mathfrak{p}_2)}.
\end{equation}
The following assertions hold true:

a) If $\mathfrak{p}_2\geq 1$ or $Q<Q_{0}$ then $R_r<R_c$ and $R_r$ is the first critical Rayleigh number. That is, there is a finite set of critical indices $\mathcal{C}$ minimizing \eqref{Rr} such that
\begin{equation}\label{PES1}
\begin{aligned}
& \beta _{J}^{1}\left( R\right) 
\left\{ 
\begin{array}{c}
>0\text{ if }R<R_{r}, \\ 
=0\text{ if }R=R_{r}, \\ 
<0\text{ if }R>R_{r}
\end{array} \right. 
&& \forall J\in\mathcal{C}, \\
& Re\beta_{J} \left( R_{r}\right) <0 &&\forall J \notin \mathcal{C}.
\end{aligned}
\end{equation}

b) If $\mathfrak{p}_2<1$ and $Q>Q_{0}$ then $R_c<R_r$ and $R_c$ is the first critical Rayleigh number. That is, there is a finite set of critical indices $\mathcal{C}$ minimizing \eqref{Rc} such that
\begin{equation} \label{PES2}
\begin{aligned}
& Re \beta _{J}^{1}(R) =Re\beta _{J}^{2}( R)
\left\{ 
\begin{array}{c}
>0\text{ if }R<R_{c}, \\ 
=0\text{ if }R=R_{c}, \\ 
<0\text{ if }R>R_{c}
\end{array}\right. 
&& \forall J\in\mathcal{C},\\
& Re\beta_J \left( R_{c}\right) <0 &&\forall J \notin \mathcal{C}.
\end{aligned}
\end{equation}

\end{proposition}
There is no simple formula for
the critical Chandrasekhar number $Q_{0}$ which separates the stationary and oscillatory convection regimes. $Q_{0}$ is determined by either
the vanishing of the oscillation frequency $\rho $ (the imaginary part of the first critical complex eigenvalues at $R=R_c$) given by%
\begin{equation}\label{rho}
\rho ^{2}=\mathfrak{p}_1\mathfrak{p}_2 [-\frac{\mathfrak{p}_2\gamma _{J_{c}}^{4}}{\mathfrak{p}_1}+%
\frac{\left( 1-\mathfrak{p}_2\right) Q\pi ^{2}}{\mathfrak{p}_1+1}], \quad J_c\in\mathcal{C} .
\end{equation}
Utilizing $\rho^2>0$, one establishes the estimate \eqref{Q0}. 
Also it is worth mentioning that in \cite{banerjee1985principle}, the authors prove that $Q_0\mathfrak{p}_2>\pi^2$.

\section{Dynamic Transitions}
With the PES condition at our disposal, the dynamic transition theorem in \cite{ptd} ensures that the system always undergoes a dynamic transition to one of the three types (Type-I, Type-II or Type-III) as the Rayleigh number crosses the first critical value $R_r$ or $R_c$. The type of transition and associated pattern formation are then determined by studying the nonlinear interactions of the problem.

We address in this section the dynamic transition and the pattern formation of (\ref{Eqn1})-(\ref{B.C.}).
Our main results will depend on two nondimensional numbers $a$ and $b$ which are defined as follows:
\begin{equation}\label{a,b}
\begin{aligned}
& a= \pi^4 Q(\pi^2-5 \alpha_J^2)-\mathfrak{p}_2^2 \alpha_J^4 R_r +\kappa_{2j_1,0}+\kappa_{0,2j_2}, \\
& b= 2\pi^4 Q(\pi^2-\alpha_J^2)-\mathfrak{p}_2^2\alpha_J^4 R_r,
\end{aligned}
\end{equation}
where $J=(j_1,j_2,1)$ denotes the critical index and $\kappa$ is given by:
\begin{equation}\label{kappa}
\begin{aligned}
&
\kappa_{s_1,s_2}:= \nu_{s_1,s_2}\big[(\mathfrak{p}_1\pi^2 Q(1+\frac{4\gamma_J^2}{\gamma_{s_1,s_2,2}^2})-\mathfrak{p}_2\gamma_J^4)(R_r+\eta_{s_1,s_2})
 -\frac{\mathfrak{p}_1\mathfrak{p}_2 R_r \alpha_J^2}{\gamma_{s_1,s_2,2}^2}(R_{s_1,s_2,2}+\eta_{s_1,s_2})\big],\\
& \eta_{s_1,s_2}:=\frac{\gamma_{s_1,s_2,2}^2}{\alpha_J^2}\left(\frac{\pi^2 Q}{\mathfrak{p}_2}+\frac{\gamma_J^4}{\mathfrak{p}_1}\right), 
\qquad \nu_{s_1,s_2}=\frac{2 \mathfrak{p}_2 \pi^2(\alpha_J^2-\frac{1}{4}\alpha_{s_1,s_2}^2)^2}{\mathfrak{p}_1 (R_{s_1,s_2,2}-R_r)\alpha_J^2},\\
& R_{s_1,s_2,s_3}=\frac{\gamma_{s_1,s_2,s_3}^2}{\alpha_{s_1,s_2}^2}(\gamma_{s_1,s_2,s_3}^4+Q (s_3 \pi)^2),  
\qquad R_r=\min R_{s_1,s_2,s_3}=\frac{\gamma_J^2}{\alpha_J^2}(\gamma_J^4+Q \pi^2).
\end{aligned}
\end{equation}

\subsection{Transitions from simple real eigenvalues}
In the case where the magnetic Prandtl number $\mathfrak{p}_2\geq1$ or the Chandrasekhar number $Q<Q_0$ where $Q_0$ depends on the system  parameters, the first eigenvalues are real. Generically speaking the first eigenvalue is simple and the dynamical transition is characterized by the following theorem.
\begin{theorem}\label{rollthm}
If $\mathfrak{p}_2\geq 1$ or $Q<Q_{0}$ then the problem (\ref
{Eqn1}) with (\ref{B.C.}) undergoes a dynamic transition at $R=R_{r}$. Consider the case where the first critical eigenvalue is simple.
\begin{itemize}
\item[a)]  When the first critical wave index is of the form $J=(0,j,1)$ or $J=(j,0,1)$ with non-zero $j$, the first transition is,
 \begin{equation*}
\begin{aligned}
& \text{Type-I, } &&\text{if }b<0,\\
& \text{Type-II, }&& \text{if } b>0.
\end{aligned}
\end{equation*}
\item[b)] When the first critical wave index is of the form $J=(j,k,1)$ with $j\neq0$, $k\neq0$. Then the first transition is,
 \begin{equation*}
\begin{aligned}
& \text{Type-I, } &&\text{if }a<0,\\
& \text{Type-II, }&& \text{if } a>0.
\end{aligned}
\end{equation*}
\item[c)] Moreover, when the transition is Type-I, then as $R$ crosses $R_{r}$, the system bifurcates to two steady state
solutions $\psi_{\pm}$ which can be expressed as
\begin{equation*}
\psi _{\pm}(R)=\pm c (R-R_r ) ^{1/2}\psi _{J}+o(( R-R_r) ^{1/2}) ,
\end{equation*}%
where $c>0$ is a constant, and $\psi _{J}$ is the critical eigenvector. Furthermore, there is an open set $U\in H$ with $\psi =0\in U$ such that $\overline{U}=%
\overline{U_{+}}\cup \overline{U_{-}}$, $U_{+}\cap U_{-}=\emptyset $, $\psi
=0\in \partial U_{+}\cap \partial U_{-}$ and $\psi _{\pm}(R)$ attracts $U_{\pm}$.

\end{itemize}
\end{theorem}
The following corollary of Theorem \ref{rollthm} gives sufficient conditions for the existence of Type-I transition when the critical mode is a roll pattern.
\begin{corollary}\label{rollcorollary}
Under the assumptions of Theorem \ref{rollthm}, when the first critical index is of the form $J=(0,j,1)$ or $J=(j,0,1)$ with non-zero $j$, the transition at $R=R_r$ is always Type-I if $Q>Q_{\ast}$ or $\mathfrak{p}_2>\mathfrak{p}_{\ast}$ where $Q_{\ast}$ and $\mathfrak{p}_{\ast}$ depends on the length scales of the domain. 
\begin{itemize}
\item[a)] Regardless of the length scales $L_1$, $L_2$ of the domain, $Q_{\ast}<307$ and $\mathfrak{p}_{\ast}<2.24$. 
\item[b)] As $\max\{L_1,L_2\}\rightarrow\infty$, we have $Q_{\ast}\rightarrow4\pi^2$ and $\mathfrak{p}_{\ast}\rightarrow2/\sqrt{3}$.
\end{itemize}
\end{corollary}

\subsection{Transitions from real eigenvalues with multiplicity two}
In this section we consider the transitions when two modes characterizing a hexagon becomes unstable at the same critical parameter. For this to happen, a certain relation has to hold for the length scales of the box:
\begin{equation} \label{hexbox}
\frac{L_1}{L_2}=\frac{j}{k \sqrt{3}}. \end{equation}
where $j$ and $k$ are two positive integers; see also \cite{dauby1993} for a discussion. With this relation in hand, the modes with indices $I=(j,k,1)$ and $J=(0,2k,1)$ have the same wave number, $\alpha_I=\alpha_J$. Moreover, the vector field $x_I \psi_I+x_J \psi_J$ has a hexagonal pattern parallel to $y$-axis when $x_I=\pm 2 x_J$; see Figure \ref{hexa}. Thus a pair of modes will be critical at the critical Rayleigh number $R_r$ whenever one of the modes minimizes the relation \eqref{Rr}. In Figure \ref{preferredhexQ10}, the red lines show when this is the case at $Q=10$. We remark here that when the length scales are in close vicinity of the relation \eqref{hexbox}, although the first critical eigenvalue may be simple, the two eigenvalues will be very close to each other at the critical parameter $R_r$ and the structure of transition will still persist.

\begin{theorem} \label{hexthm}
Consider the case where  $I=(j,k,1)$ and $J=(0,2k,1)$ are the first critical indices which satisfy the relation \eqref{hexbox}. Consider the eight regions in the $a$--$b$ plane as shown in given by Figure \ref{regions2} with $a$ and $b$ as given in \eqref{a,b}. If $\mathfrak{p}_2\geq 1$ or $Q<Q_0$ then the first dynamic transition occurs at $R=R_r$. Moreover the following statements hold true.
\begin{itemize}
\item[i)] When $(a,b)$ is in one of the regions $\mathcal{I}_1$, $\mathcal{I}_2$,  the transition is Type-I. The topological structure of the transition is given in Figure \ref{VI} and Figure \ref{VII} respectively. In both regions, there is an attractor bifurcating on $R>R_r$ which is homeomorphic to $S^1$ containing eight bifurcated points and the heteroclinic orbits containing them. Four of these points are attractors and four of them are saddles. In region $\mathcal{I}_1$, hexagons $2\phi_I+\phi_J$ are the minimal attractors, and in region $\mathcal{I}_2$, rolls $\phi_J$ and rectangles $\phi_I$ are the minimal attractors.
\item[ii)]When $(a,b)$ is in one of the regions $\mathcal{II}_1,\,\mathcal{II}_2,\,\mathcal{II}_3,\,\mathcal{II}_4$ and $\mathcal{II}_5$, the transition is Type-II. The topological structure of the transition is given in Figure \ref{I}--Figure \ref{V} respectively. In regions $\mathcal{II}_2$ and $\mathcal{II}_3$, there is a repeller bifurcating on $R<R_r$ which is homeomorphic to $S^1$ and no bifurcated states on $R>R_r$. In the regions $\mathcal{II}_1,\,\mathcal{II}_4,\,\mathcal{II}_5$ there are bifurcated points on both $R>R_r$ and $R<R_r$. 
\item[iii)] When $(a,b)$ is in region $\mathcal{III}_1$ the transition is Type-III; see Figure \ref{VIII}. There is a neighborhood $\mathcal{U}$ of $\psi=0$ in $H$ such that for any $R_r<R<R_r+\epsilon$ with some $\epsilon>0$, $\mathcal{U}$ can be decomposed into two open sets $\mathcal{U}_1^{R}$, $\mathcal{U}_2^{R}$,
$$\overline{\mathcal{U}}=\overline{\mathcal{U}_1^{R}}\cup\overline{\mathcal{U}_2^{R}}, \qquad \mathcal{U}_1^{R}\cap\mathcal{U}_2^{R}=\emptyset$$ such that
\begin{equation*}
\begin{aligned}
& \lim_{R\rightarrow R_r}\limsup_{t\rightarrow\infty} ||S_{R}(t,\psi)||_H=0 && \qquad \forall \psi\in\mathcal{U}_1^{R}, \\
& \limsup_{t\rightarrow\infty} ||S_{R}(t,\psi)||_H\geq\delta>0&& \qquad \forall \psi\in\mathcal{U}_2^{R},
\end{aligned}
\end{equation*}
for some $\delta>0$. Here $S_{R}$ is the evolution of the solution with initial data $\psi$.
Moreover $\mathcal{P}(U_1^R)$ is a small perturbation of the union of two sectorial regions which are given by:
\begin{equation}
\begin{aligned}
& \mathcal{P}(\mathcal{U}_{1,+}^R)\approx\mathcal{U}\cap \{x\in \mathbb{R}^2\mid -\tan^{-1}{1/2}<arg(x)<\tan^{-1}{1/2}\},\\
& \mathcal{P}(\mathcal{U}_{1,-}^R)\approx\mathcal{U}\cap\{x\in \mathbb{R}^2\mid \pi-\tan^{-1}{1/2}<arg(x)<\pi+\tan^{-1}{1/2}\},
\end{aligned}
\end{equation}
where $\mathcal{P}$ is the projection onto $\phi_{I}$, $\phi_{J}$ plane.

Moreover, the system bifurcates to two attractors with basin of attraction $\mathcal{U}_{1,+}^R$, $\mathcal{U}_{1,-}^R$. Each attractor consists of three steady states and the connecting heteroclinic orbits. The steady states $\phi_I$ having rectangular patterns are the minimal attractors of the bifurcated attractor. 
\end{itemize}
\end{theorem}

\begin{corollary}\label{hexcor}
Under the assumptions of Theorem \ref{hexthm}, when $\mathfrak{p}_2\geq8$, the transition at $R=R_r$ is always Type-I. Moreover, there is an attractor bifurcating on $R>R_r$, which is homeomorphic to $S^1$. This attractor contains eight steady states and the orbits connecting them; see Figure \ref{VI}. Two of the steady states have roll patterns, two of them have rectangular patterns and four of them have hexagonal patterns. The rolls and rectangles are the minimal attractors of this attractor while the hexagons are unstable states after the transition.
\end{corollary}
\begin{figure}
\includegraphics[scale=0.6]{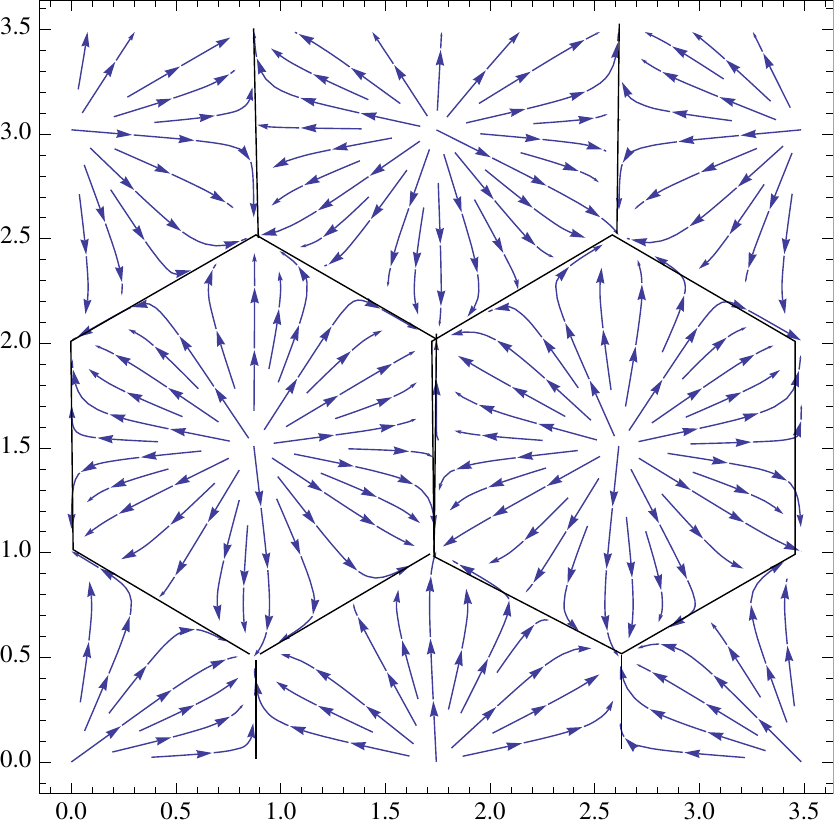}
\caption{The flow structure of $2\phi_{2,1}+\phi_{0,2}$ at $z=1$ for the box dimensions $L_1=L_2=3.5$}
\label{hexa}
\end{figure}

\begin{figure}
\subfigure[]{
\includegraphics[scale=0.45]{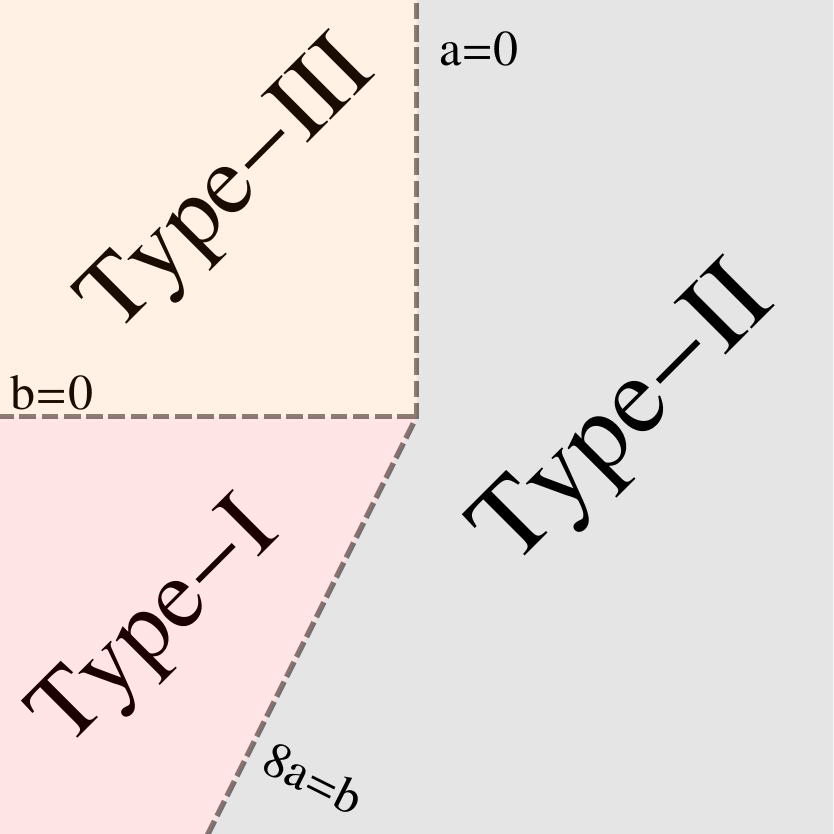}
}
\subfigure[]{
\includegraphics[scale=0.6]{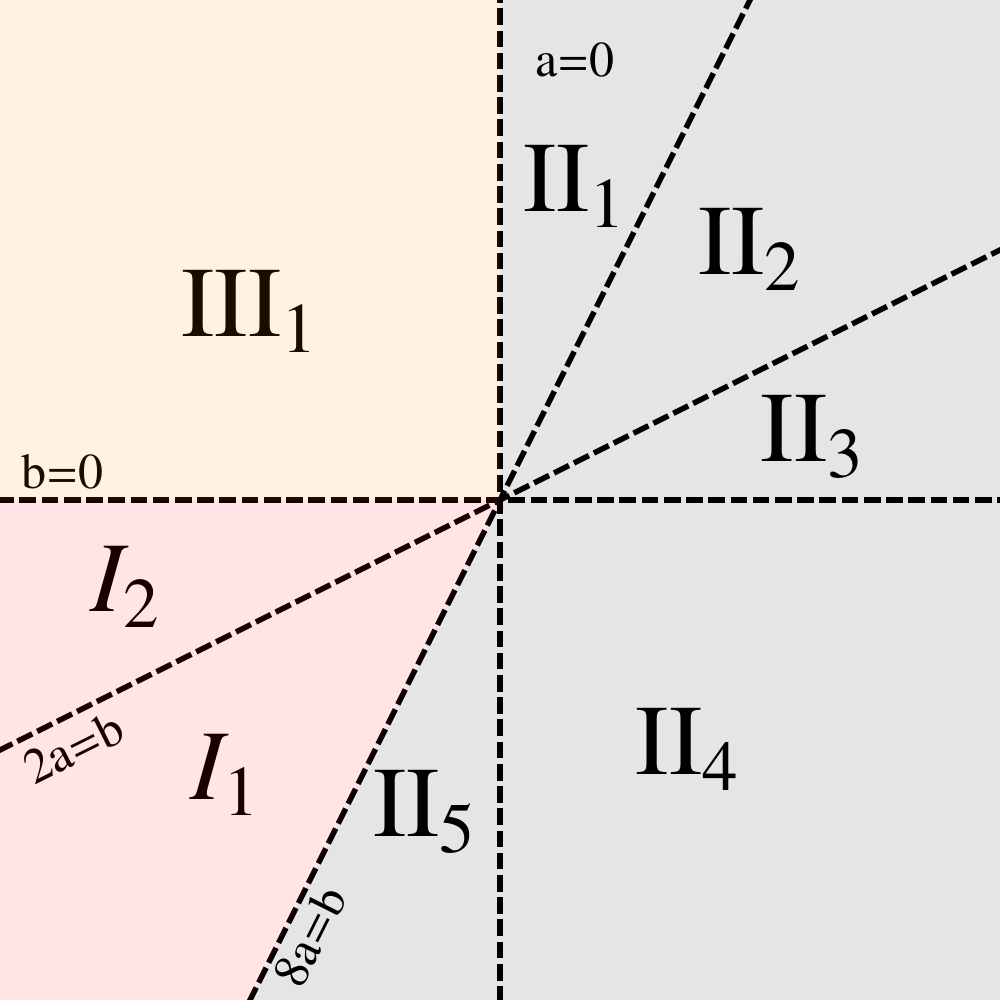}
}
\caption{a) Regions of different type of transitions in the coefficient a-b plane and b) The topological structure of the solutions in different regions are shown in Figure \ref{I} through Figure \ref{VIII}. }
\label{regions2}
\end{figure}

\begin{figure}
\centering
\subfigure[$R\leq R_r$]{
\includegraphics[scale=0.4]{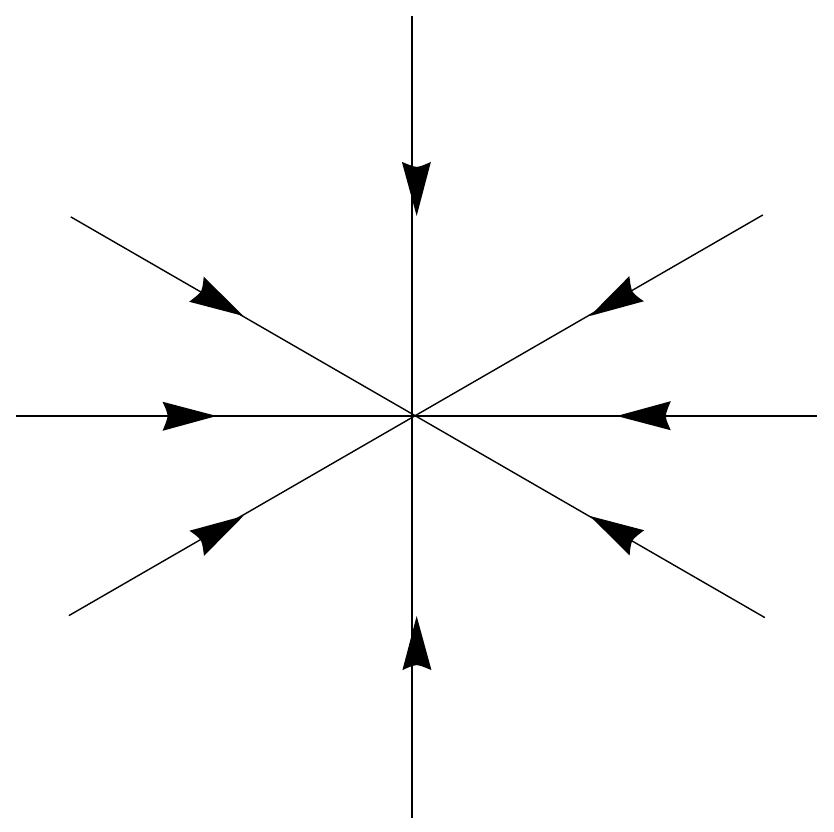}
}
\subfigure[$R>R_r$]{
\includegraphics[scale=0.4]{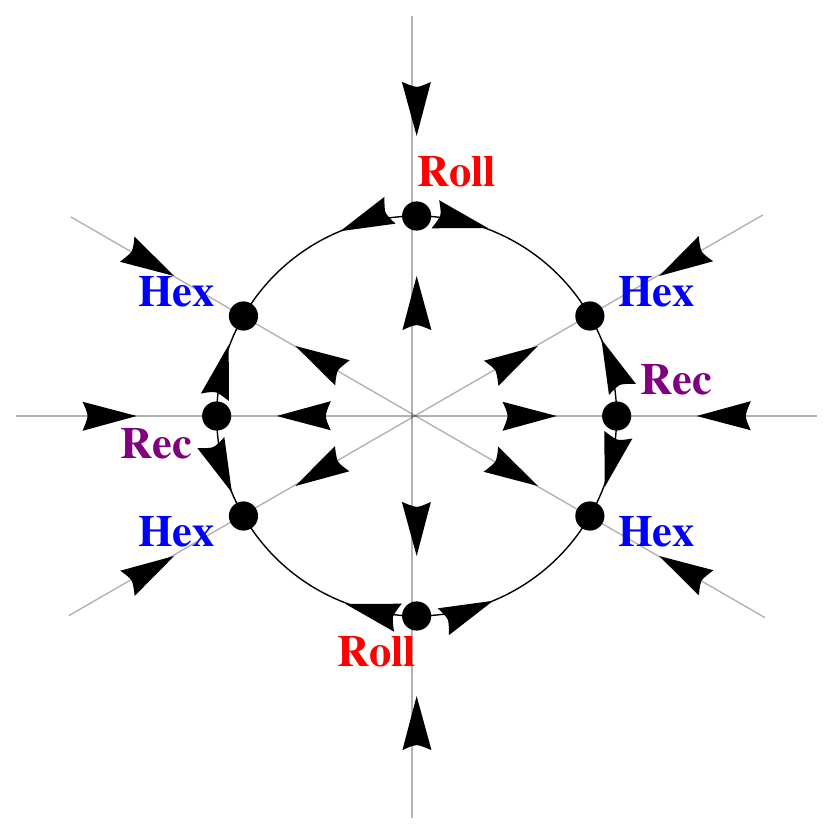}
}
\caption{Region $\mathcal{I}_1$.}
\label{VI}
\end{figure}

\begin{figure}
\centering
\subfigure[$R\leq R_r$]{
\includegraphics[scale=0.4]{VI1.pdf}
}
\subfigure[$R>R_r$]{
\includegraphics[scale=0.4]{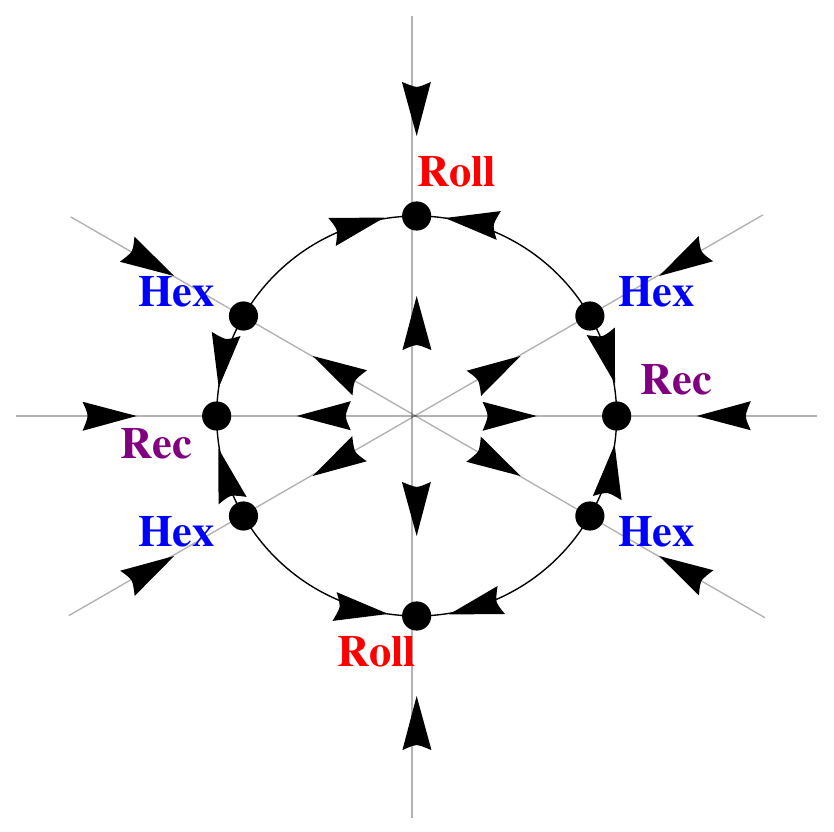}
}
\caption{Region $\mathcal{I}_2$.}
\label{VII}
\end{figure}

\begin{figure}
\centering
\subfigure[$R<R_r$]{
\includegraphics[scale=0.4]{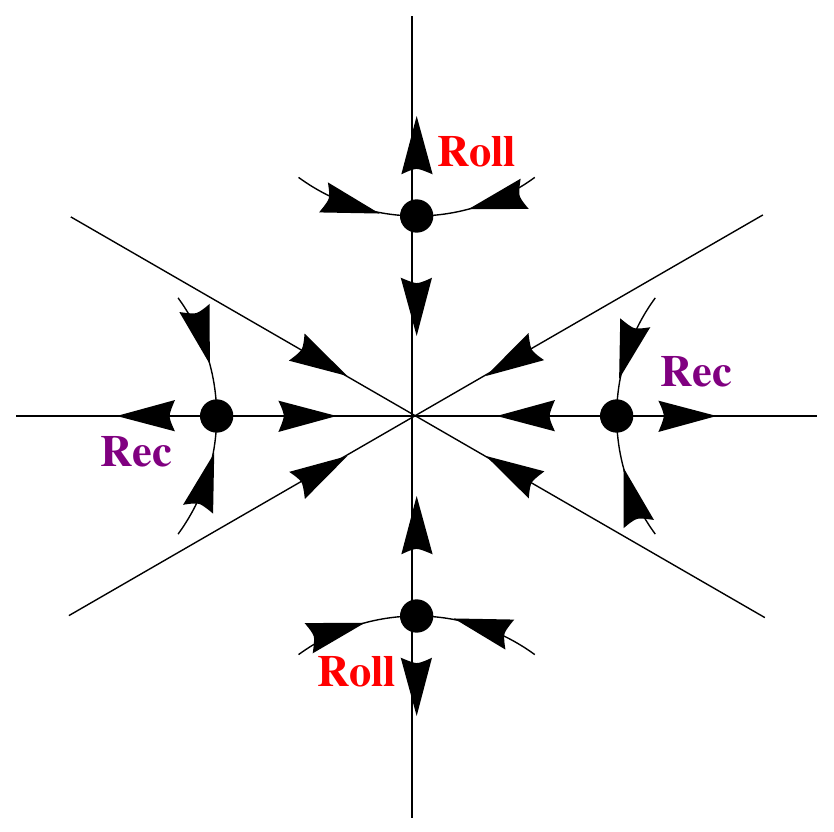}
}
\subfigure[$R=R_r$]{
\includegraphics[scale=0.4]{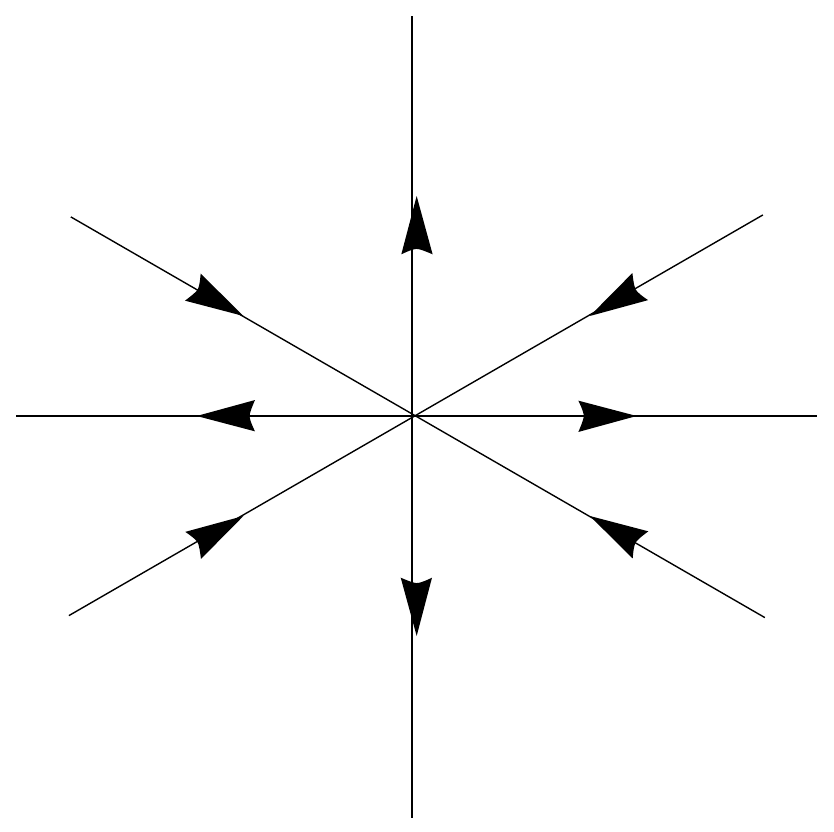}
}
\subfigure[$R>R_r$]{
\includegraphics[scale=0.4]{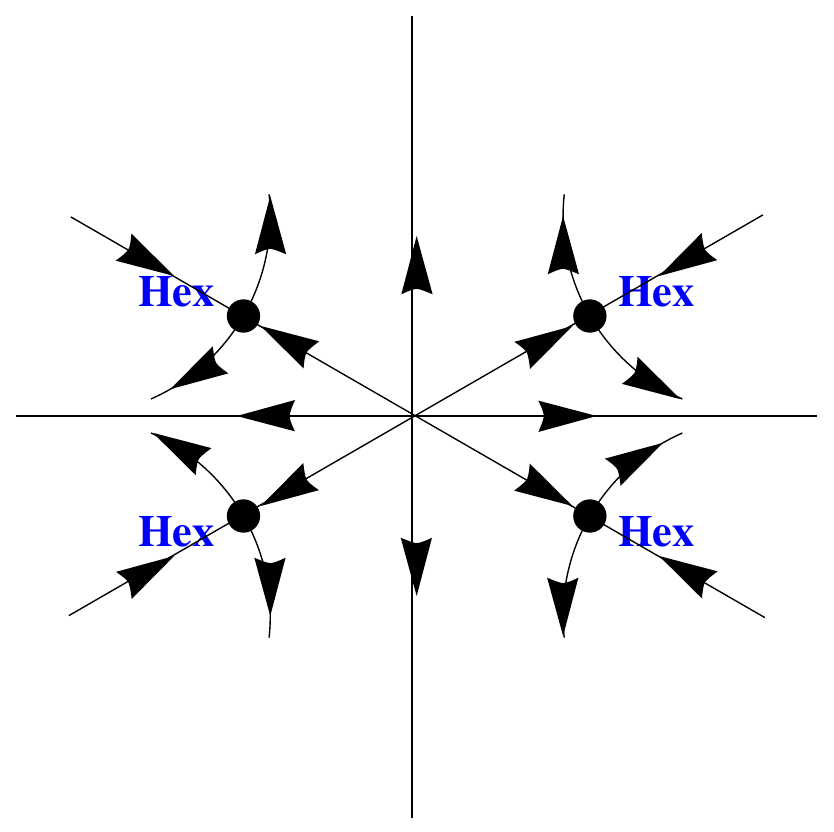}
}
\caption{Region $\mathcal{II}_1$.}
\label{I}
\end{figure}

\begin{figure}
\centering
\subfigure[$R<R_r$]{
\includegraphics[scale=0.4]{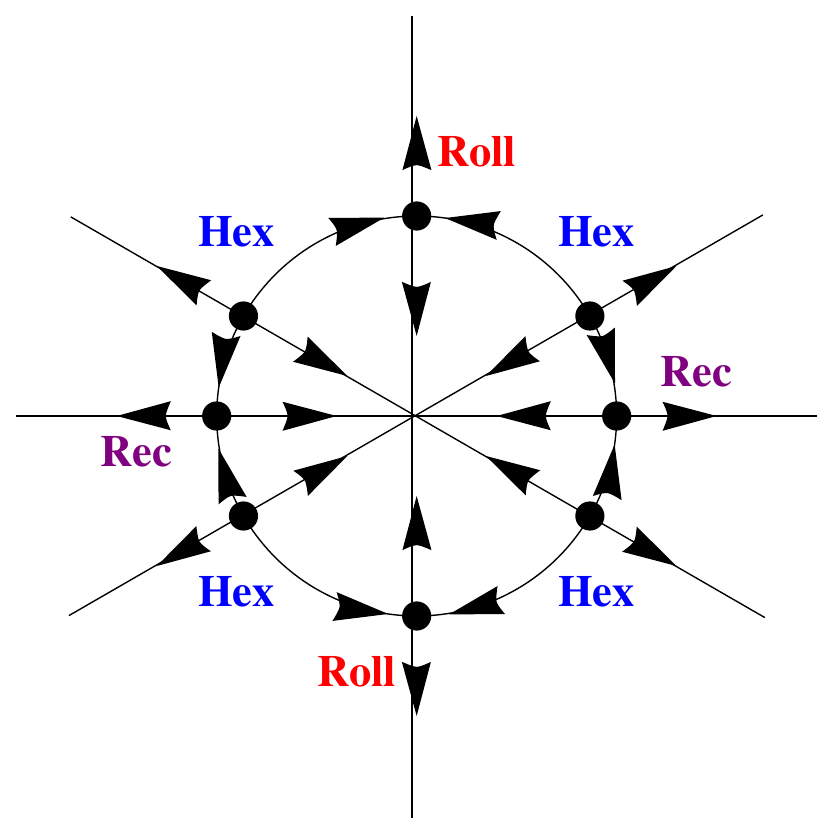}
}
\subfigure[$R\geq R_r$]{
\includegraphics[scale=0.4]{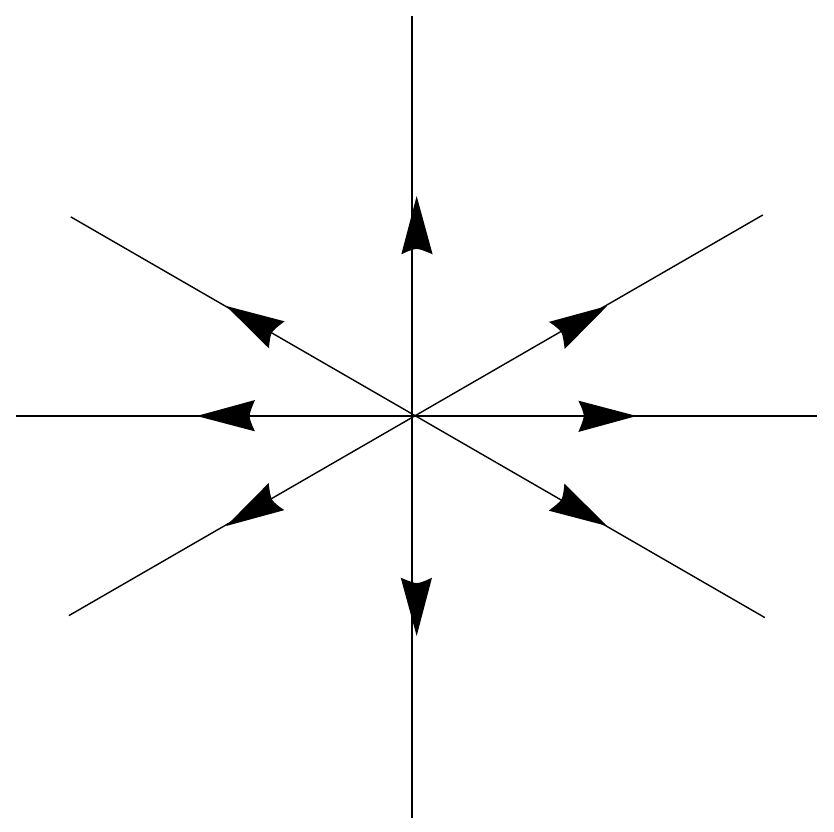}
}

\caption{Region $\mathcal{II}_2$.}
\label{II}
\end{figure}

\begin{figure}
\centering
\subfigure[$R<R_r$]{
\includegraphics[scale=0.4]{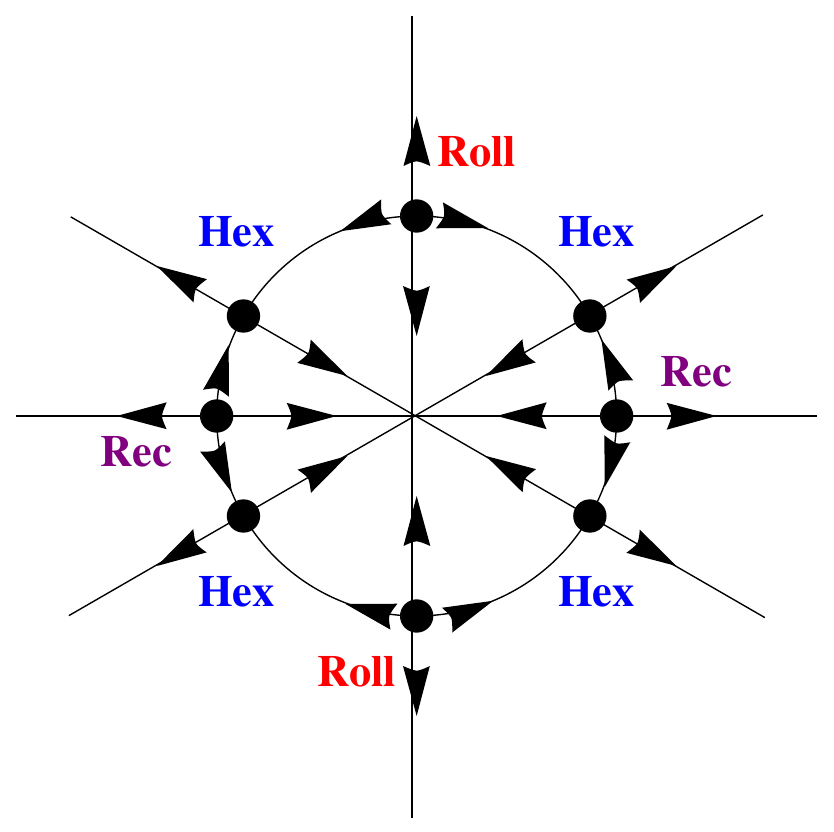}
}
\subfigure[$R\geq R_r$]{
\includegraphics[scale=0.4]{II2.pdf}
}

\caption{Region $\mathcal{II}_3$.}
\label{III}
\end{figure}

\begin{figure}
\centering
\subfigure[$R<R_r$]{
\includegraphics[scale=0.4]{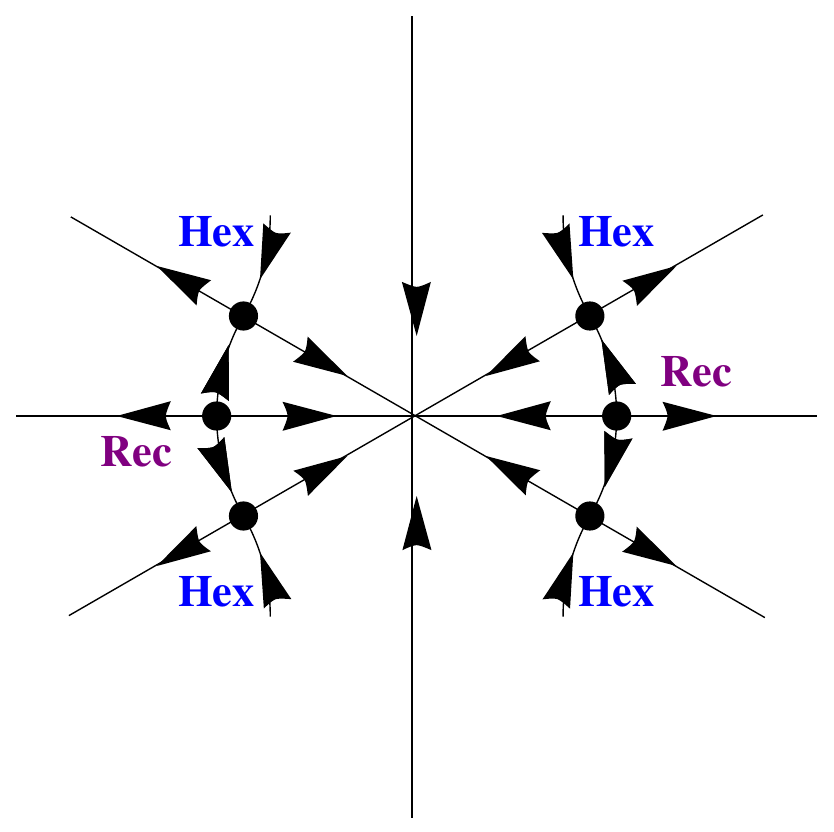}
}
\subfigure[$R=R_r$]{
\includegraphics[scale=0.4]{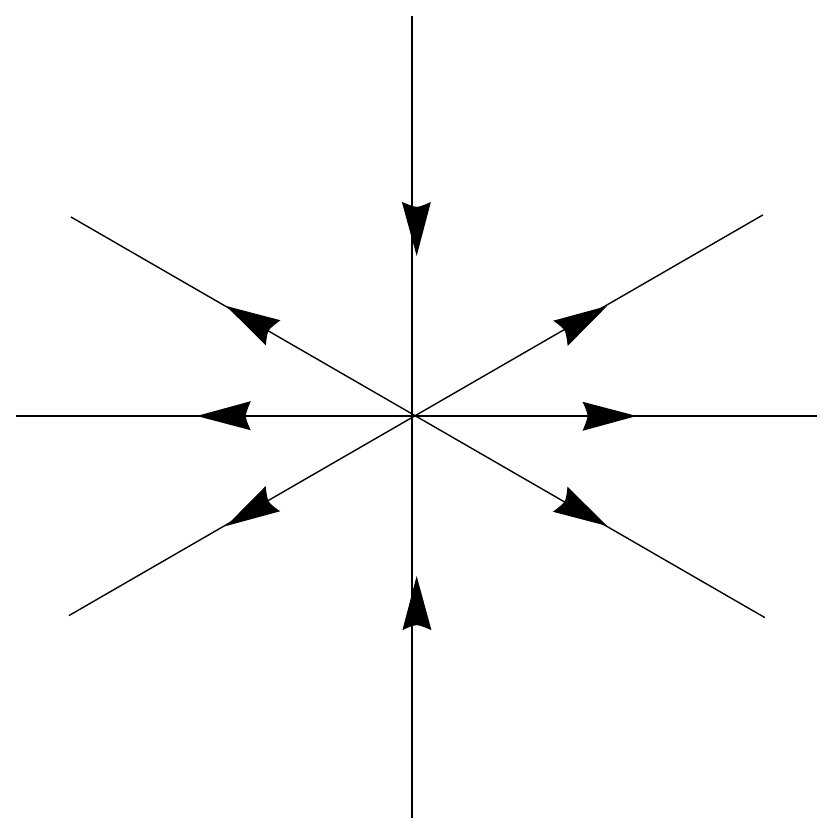}
}
\subfigure[$R>R_r$]{
\includegraphics[scale=0.4]{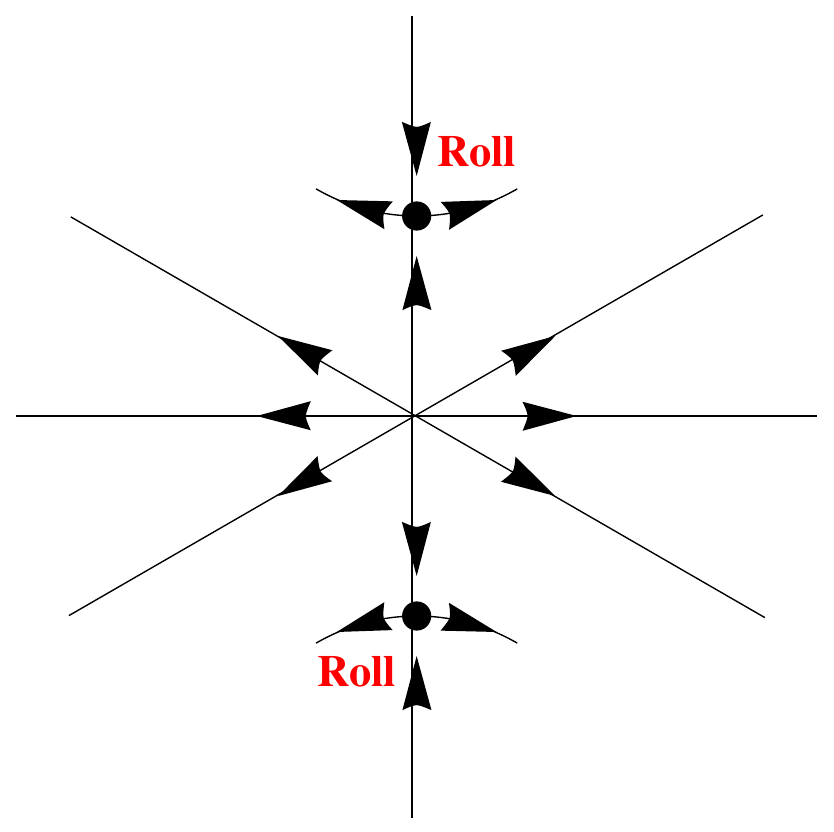}
}
\caption{Region $\mathcal{II}_4$.}
\label{IV}
\end{figure}

\begin{figure}
\centering
\subfigure[$R<R_r$]{
\includegraphics[scale=0.4]{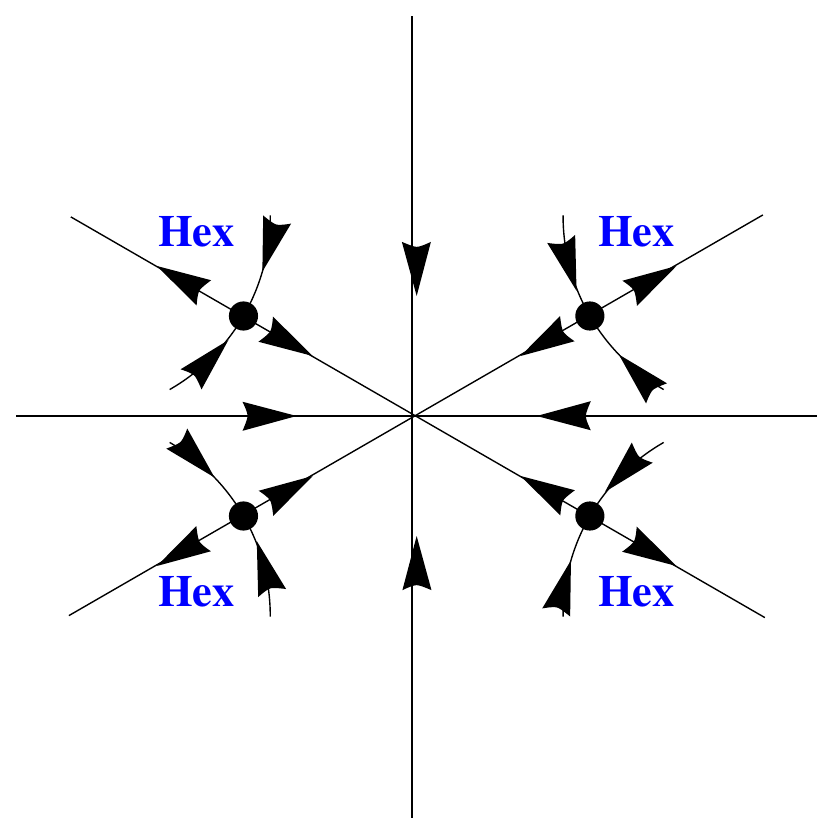}
}
\subfigure[$R=R_r$]{
\includegraphics[scale=0.4]{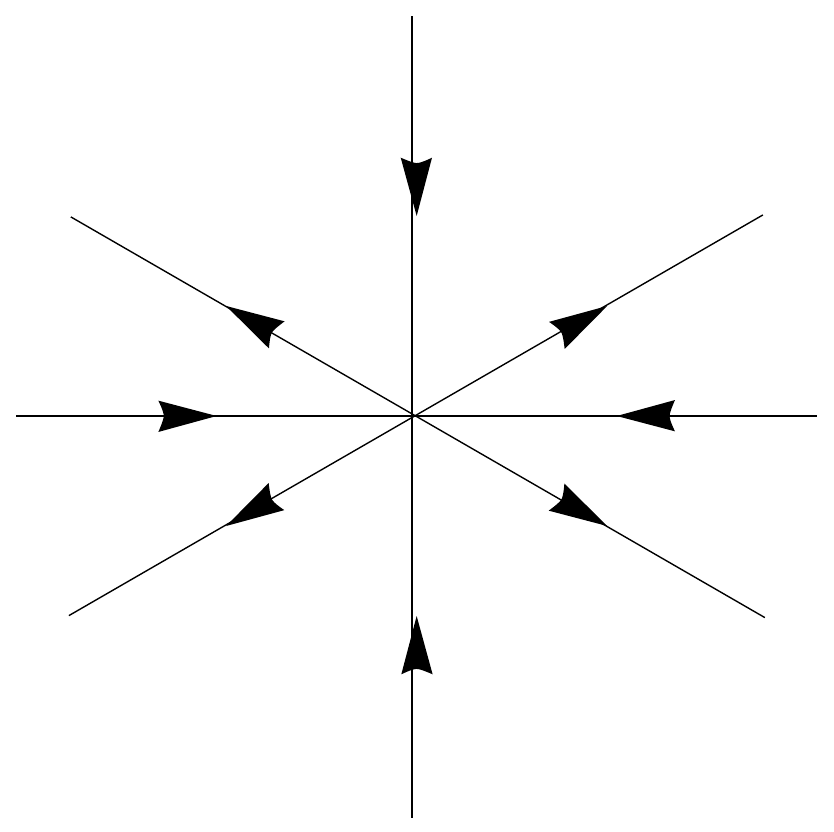}
}
\subfigure[$R>R_r$]{
\includegraphics[scale=0.4]{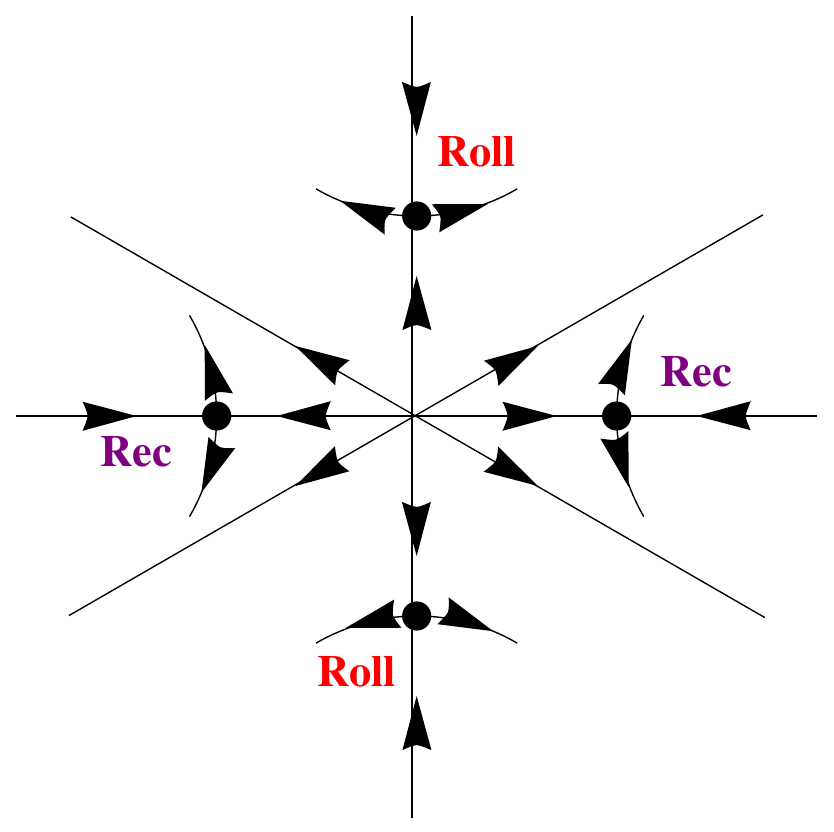}
}
\caption{Region $\mathcal{II}_5$.}
\label{V}
\end{figure}

\begin{figure}
\centering
\subfigure[$R<R_r$]{
\includegraphics[scale=0.4]{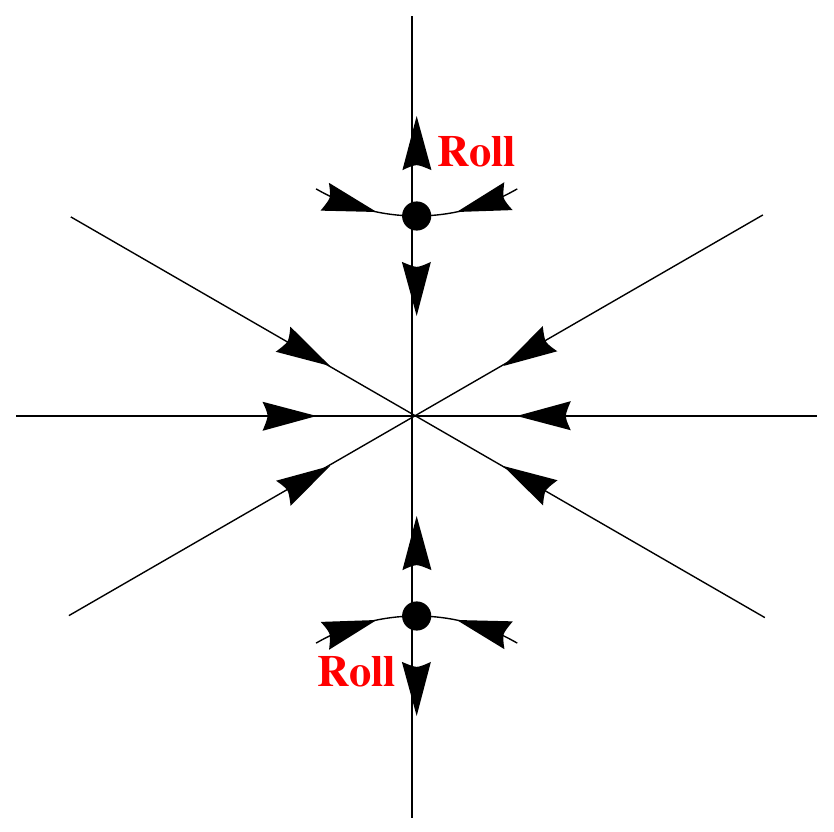}
}
\subfigure[$R=R_r$]{
\includegraphics[scale=0.4]{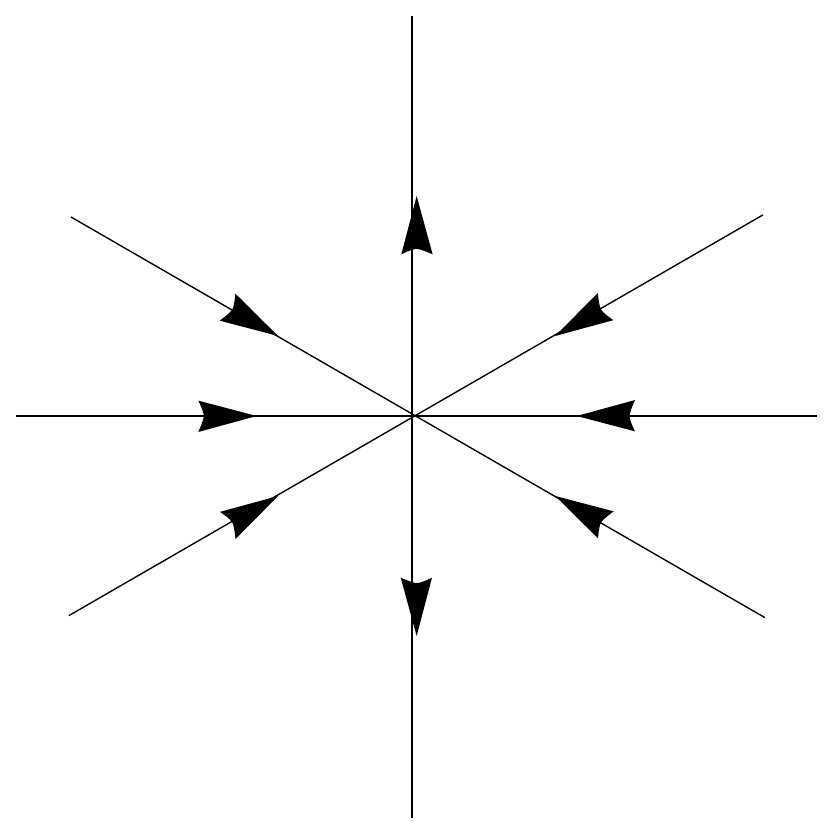}
}
\subfigure[$R>R_r$]{
\includegraphics[scale=0.4]{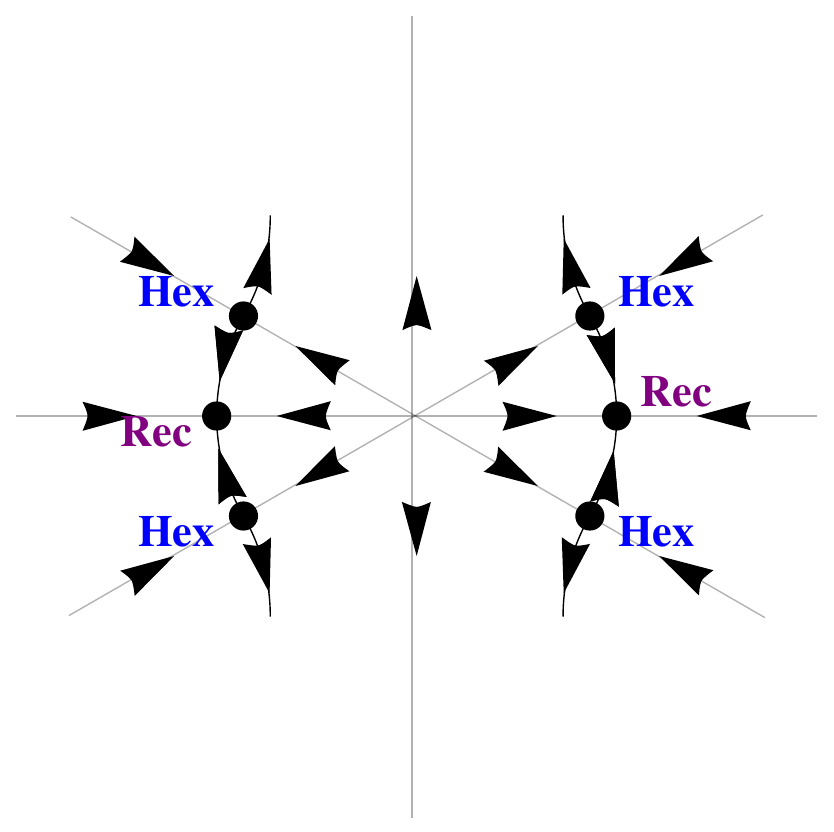}
}
\caption{Region $\mathcal{III}_1$.}
\label{VIII}
\end{figure}

\begin{remark}
To illustrate the conclusions of Theorem \ref{hexthm} we will consider two cases. Solving $b$ for $\mathfrak{p}_2$, we find that $b<0$ (resp. $b>0$) if $\mathfrak{p}_2^2>\sigma_{roll}$ (resp. $\mathfrak{p}_2^2<\sigma_{roll}$) where
\begin{equation} \label{sigmaroll}
\sigma_{roll}=\frac{2\pi^4(\pi^2-\alpha_J^2)Q}{\alpha_J^2\gamma_J^2(\pi^2 Q+\gamma_J^4)}.
\end{equation}First we will consider the case of small length scales $L_1$, $L_2$ such that the conditions \eqref{hexbox} is satisfied. Figure \ref{preferredhexQ10} shows some of the length scales for which this happens at $Q=10$.

\begin{figure}
\includegraphics[scale=.4]{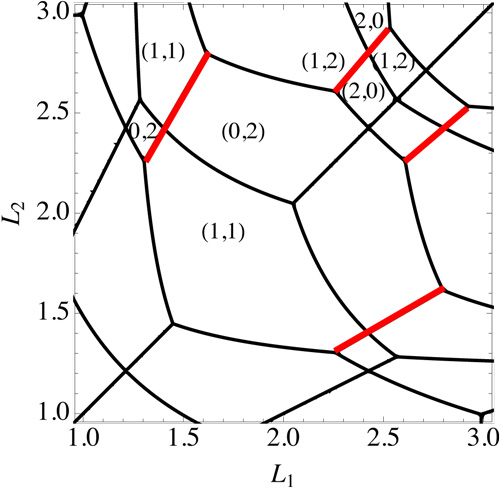}
\caption{The bold lines indicate where two modes satisfying \eqref{hexbox} are critical at $Q=10$.}
\label{preferredhexQ10}
\end{figure}

As an example we choose 
\begin{equation}\label{ex1}
Q=10,\quad L_1=3/2, \quad L_2=\sqrt{3}L_1,
\end{equation}
 Under these conditions, two modes satisfying \eqref{hexbox} will be critical with indices $I=(1,1,1)$ and $J=(0,2,1)$. The critical wave number is $\alpha_J=2\pi/L_2\approx2.42$. In this case $\sqrt{\sigma_{roll}}\approx0.5$. Calculating $a/b$, we find that under the conditions \eqref{ex1}, the system is in region $\mathcal{III}_1$ for $\mathfrak{p}_2<\sqrt{\sigma_{roll}}$ and will move into $\mathcal{I}_2$ as $\mathfrak{p}_2$ crosses $\sqrt{\sigma_{roll}}$.

Next we consider the limit of large length scales $L_1>>1$, $L_2=\sqrt{3}L_1$. In this case, the critical wave number is given by \eqref{Qalpharoll}. If we take $Q=10$ again, then we find $\alpha_J=2.59$. Once again, the system moves from region $\mathcal{III}_1$ to region $\mathcal{I}_2$, this time at $\mathfrak{p}_2=\sqrt{\sigma_{roll}}\approx0.39$.

\end{remark}

\subsection{Transitions from complex simple eigenvalues}

Under the conditions $\mathfrak{p}_2<1$ and $Q>Q_0$, the first critical eigenvalues will be a pair of complex numbers. To demonstrate the main ideas and for simplicity, we only consider the transition from a simple complex eigenvalue with index $J_c=(j_c,0,1)$ or $J_c=(0,k_c,1)$ for $j_c,k_c\geq1$.
\begin{theorem}
\label{Complex MHD}Consider $b$ which is defined by \eqref{b1}. Assume $\mathfrak{p}_2<1$ and $Q>Q_{0}$ and assume that the critical index $J_c=(j_c,0,1)$ or $J_c=(0,k_c,1)$ then we have the
following assertions:

\begin{enumerate}
\item If $b<0$ then the problem undergoes a Type-I transition at $R=R_{c}$
and bifurcates to a periodic solution on $R>R_{c}$ which is an attractor,
the periodic solution can be expressed as%
\begin{align*}
& \psi =x\left( t\right) Re\psi _{J_{c}}+y\left( t\right) Im\psi
_{J_{c}}+o\left( \left\vert \alpha \left( \lambda \right) \right\vert
^{1/2}\right) , \\
& x\left( t\right) =\left( \frac{\alpha \left( \lambda \right) }{\left\vert
b\right\vert }\right) ^{1/2}\cos \rho t, \\
& y\left( t\right) =\left( \frac{\alpha \left( \lambda \right) }{\left\vert
b\right\vert }\right) ^{1/2}\sin \rho t.
\end{align*}
In particular, if $Q$ is sufficiently large or $\rho$ is sufficiently small then the transition is Type-I.

\item If $b>0$ then the transition is a jump transition and the system
bifurcates on $R<R_{c}$ to a unique unstable periodic orbit.
\end{enumerate}
\end{theorem}

\section{Proof of the Main Theorems}

\subsection{Reduction strategy}
In dynamic transition problems, the method of proof relies heavily on the reduction of the problem to the center
manifold in the first unstable eigendirections. The center manifold is usually expanded by the eigenmodes of the linear operator of the problem and the coefficients of the the expansion are found by the approximation formulas derived in \cite{ptd}. However, for our problem, such an analysis is not suitable because some eigenfunctions of the linear operator that are needed to compute the nonlinear interactions have complicated forms. In particular the eigenvalues of \eqref{Lin-EV-Cpct} corresponding to indices $(j,k,l)$ with $l\geq2$, $j^2+k^2\neq0$ can only be found by first computing the roots of \eqref{polynom}. It is easy to see that these roots will have complicated forms involving the parameters of the system. This makes it very difficult to compute their contributions in the nonlinear interactions. We will overcome this difficulty by expanding the center manifold in terms of the eigenfunctions of the Laplacian operator with the same boundary conditions. 

In the case where the first critical eigenvalues are real, the critical eigenfunctions can be found by setting $\beta=0$ in \eqref{Lin-EV-Cpct}. Thanks to the boundary conditions, the eigenvectors can be expressed using the separation of variables. For example the vertical component of the velocity field becomes:
\begin{equation}
\label{seperationofvar}
\begin{aligned}
& w =W_J\cos \frac{j_1 \pi  x_{1}}{L_1}\cos \frac{j_2 \pi x_2}{L_2}\sin j_3\pi x_{3},  \\
\end{aligned}
\end{equation}
and analogous expressions hold for the other components of the solutions of the linear equations. Here $J=\left( j_1,j_2,j_3\right) $ with nonnegative integers $j_1,j_2,j_3$. It is clear that critical indices must have $j_3=1$ as the minimum in \eqref{Rr} is achieved for $j_3=1$ and moreover we must have $\alpha_J>0$.
When $\alpha_J\neq 0$, the amplitudes $U_J$, $V_J$ of the horizontal velocity field and $H_{1,J}$, $H_{2,J}$ of the horizontal magnetic field depend on the vertical components $W_J$ and $H_{3,J}$ as follows:
\begin{equation}\label{horizontals}
\begin{aligned}
& U_J=-\frac{j_1 \pi}{L_1 \alpha_J^2}j_3  \pi W_J,\qquad V_J=-\frac{j_2 \pi}{L_2 \alpha_J^2}j_3 \pi W_J,\\
& H_{1,J}=\frac{j_1 \pi}{L_1 \alpha_J^2}j_3  \pi H_{3,J},\qquad H_{2,J}=\frac{j_2 \pi}{L_2 \alpha_J^2}j_3 \pi H_{3,J}.\\
\end{aligned}
\end{equation}
Thus the critical eigenfunctions are determined by three quantities: the amplitude $W_J$ of the vertical velocity field, the amplitude $H_{3,J}$ of the vertical magnetic field and the amplitude $\Theta_J$ of the temperature field. These can be found as:
\[
W_J=\gamma_J^2, \qquad \Theta_J=1, \qquad H_{3,J}=\pi.
\]
Next we need to analyze the conjugate critical eigenfunctions which are the solutions of the conjugate eigenvalue problem at $\beta=0$:
\begin{equation}  \label{Conj-1}
\begin{aligned} &\mathfrak{p}_1\left( \Delta u^{\ast }-\nabla p^{\ast }\right)
-\mathfrak{p}_2\frac{\partial H^{\ast }}{\partial x_{3}}+T^{\ast
}\overrightarrow{k} =0, \\ &\Delta T^{\ast
}+R\mathfrak{p}_1u_{3}^{\ast } =0, \\ &\mathfrak{p}_2\Delta
H^{\ast }-Q\mathfrak{p}_1\frac{\partial u^{\ast }}{\partial x_{3}} =0, \\ & \text{div}u^{\ast } =\text{div}H^{\ast }=0. \end{aligned}
\end{equation}
The conjugate eigenfunctions can be found as:
\[
W_J^{\ast}=\mathfrak{p}_2 \gamma_J^2, \qquad \Theta_J^{\ast}=\mathfrak{p}_2 \mathfrak{p}_1 R_r \qquad H_{3,J}^{\ast}=-\mathfrak{p}_1 Q \pi.
\]
Now we turn to eigenvalue problem of the Laplacian operator:
\[ \Delta e_J=-\gamma_J^2 e_J. \]
The Laplacian operator with boundary conditions \eqref{B.C.} has a complete set of orthogonal eigenvectors $e_J=(\vec{u}_J,T_J,\vec{H}_J)$ which can be expressed by the same separation of variables \eqref{seperationofvar}:
  
Again there are only three independent unknowns, namely the amplitudes of the vertical velocity $w$, the temperature $T$ and the vertical magnetic field $H_3$ which are denoted by $W$, $\Theta$ and $H_3$ respectively.
 \begin{itemize}
 \item If $j_1^2+j_2^2=0$ and $j_3\neq0$, then
 \begin{equation}\label{ev1}
 e_J=(0,T_J,0), \qquad \Theta_J=1.
 \end{equation}
 \item If $j_1^2+j_2^2\neq 0$ and $j_3=0$, then
 \begin{equation}\label{ev2}
 e_J=(0,0,\vec{H}_J), \qquad H_{3,J}=1.
 \end{equation}
In this case, although there are eigenmodes which have the form $e=(\vec{u},0,0)$ with $\vec{u}=(u,v,0)$, it can be checked that these modes do not affect the calculations of the nonlinear interactions as they will be not be present in the lowest order approximation of the center manifold function. 
\item  If $j_1^2+j_2^2\neq 0$ and $j_3\neq0$, then the multiplicity of an eigenvalue is three and the eigenvectors are:
 \begin{equation}\label{ev3}
\begin{aligned}
&e_J^1=(\vec{u}_J,0,0), \qquad W_J=1,\\
&e_J^2=(0,T_J,0), \qquad \Theta_J=1,\\
&e_J^3=(0,0,\vec{H}_J), \qquad H_{3,J}=1.
\end{aligned}
\end{equation}
\end{itemize}

Proposition \ref{PrExSt} gives a critical set of indices $\mathcal{C}=\{J_1,J_2,\dots,J_k\}$ such that the eigenvalues having these indices are the first critical eigenvalues.         The dynamics on the center manifold can be captured as follows. We first write:
$$
\psi=\psi^c+\Phi(x), \qquad \psi^c=\sum_{J\in\mathcal{C}}x_J\psi_J,
$$
where $\Phi$ is the center manifold function and $x_J=x_J(t)\in\mathbb{R}$ are the amplitudes of the critical modes.
 Multiplying the governing evolution equations by the conjugate eigenvectors $\psi_J^{\ast}$, $J\in\mathcal{C}$ we get:
\begin{equation}\label{red1}
\frac{dx_J}{dt}=\beta_J(R) x_J+\frac{(G(\psi,\psi),\psi_J^{\ast})}{(\psi_J,\psi_J^{\ast})}.
\end{equation}
When the linear part of \eqref{red1} is diagonal, we have the following approximation of the center manifold; (see \cite{ptd}):
\begin{equation}\label{cmrealformula}
-\mathcal{L}_{R}\Phi \left( x,R\right) =P_{2}G\left( \psi
^c\right) +o(2),
\end{equation}
where $\mathcal{L}_{R}=L_{R}\mid _{E_{2}}$, $P_2$ is the projection onto the span of the critical eigenvectors and 
\begin{equation*}
o(n)=o\left( \left\vert x\right\vert ^{n}\right) +O\left( \left\vert
x\right\vert ^{n}\left\vert \beta\left( R\right) \right\vert
\right) .
\end{equation*}%
As we have said, we will utilize an expansion of the center manifold using eigenfunctions $e_J$ of the Laplacian operator:
\[
\Phi=\sum_S  \Phi_S e_S.
\]  The crucial point here is that $P_2 e_J=0$ if $J\in\mathcal{C}$ . This is because the eigenvectors $\phi_J$ of the original linear operator and $e_J$'s have the same form with different coefficients due to separation of variables. So
$$\langle \phi_J,e_K\rangle=0, \quad \text{for $J\neq K$}.$$

Now, to calculate this approximation we need to compute terms of the form $\langle e_S,\mathcal{L_R}^{\ast}e_S\rangle$. If we denote $S=(s_1,s_2,s_3)$, then for $S\notin\mathcal{C}$, there are three cases to consider:
\begin{itemize}
 \item If $s_1^2+s_2^2=0$ and $s_3\neq0$, then
 \begin{equation*}
 \langle e_S,\mathcal{L}_R^{\ast} e_S\rangle=-s_3^2 \pi^2\langle e_S,e_S\rangle.
 \end{equation*}
Using \eqref{cmrealformula}, the $e_S$ coefficient of the center manifold is:
 \begin{equation} \label{CM1}
 \Phi_S=\frac{\langle G(\psi^c),e_S\rangle}{s_3^2\pi^2 \langle e_S,e_S\rangle}+o(2).
 \end{equation}
 
 \item If $s_1^2+s_2^2\neq 0$ and $s_3=0$, then
\[
 \langle e_S,\mathcal{L}_R^{\ast} e_S\rangle=-\mathfrak{p}_2 \alpha_S^2\langle e_S,e_S\rangle.
\]
Using \eqref{cmrealformula}, the $e_S$ coefficient of the center manifold is:
  \begin{equation} \label{CM2}
 \Phi_S=\frac{\langle G(\psi^c),e_S\rangle}{\mathfrak{p}_2 \alpha_S^2\langle e_S^2,e_S^2\rangle}+o(2).
 \end{equation}
\item  If $s_1^2+s_2^2\neq 0$ and $s_3\neq0$, then we define:
 \begin{equation}
 \langle e_S^m,\mathcal{L}_R^{\ast} e_S^n\rangle=v_S\,A_{S}^{m,n}, \qquad m,n=1,2,3,
 \end{equation}
$A_S$ can be found by plugging the eigenvectors \eqref{ev3} into \eqref{Conj-1}:
 \renewcommand\arraystretch{1.5}
\renewcommand\tabcolsep{3pt}
 \[ A_S=\left[ \begin{array}{ccc}
-\mathfrak{p}_1\,\gamma_S^4\,\alpha_S^{-2} & R \mathfrak{p}_1 & -Q \mathfrak{p}_1\,\gamma_S^2 s_3\pi\alpha_S^{-2} \\
1 & -\gamma_S^2 & 0 \\
 \mathfrak{p}_2\,\gamma_S^2\,s_3 \pi\alpha_S^{-2} & 0 &- \mathfrak{p}_2\,\gamma_S^4\,\alpha_S^{-2} \end{array} \right].\] 
Here
\begin{equation*}
v_S=\int_{0}^{L_1}\int_{0}^{L_2}\int_{0}^{1}\cos^2\frac{s_1\pi x_1}{L_1}\cos^2\frac{s_2\pi x_2}{L_2}\cos^2s_3\pi x_3.
\end{equation*}
Clearly $v_S=\frac{L_1L_2}{2^n}$ where $n$ is the number of sub-indices of $S=(s_1,s_2,s_3)$ which are nonzero.
In this case, denoting 
\[\Phi_S=\left[ \begin{array}{c} \Phi_S^1 \quad \Phi_S^2 \quad \Phi_S^3\end{array}\right]^T, \]
where $\Phi_S^i$ denotes the coefficient of $e_S^i$, we have:
\begin{equation}\label{CM3}
v_S A_S \Phi_S= \left[ \begin{array}{c} \langle G(\psi^c),e_S^1\rangle \quad  \langle G(\psi^c),e_S^2\rangle \quad \langle G(\psi^c),e_S^3\rangle \end{array}\right]^T.
\end{equation}
By \eqref{Rr}, we see that for $S\notin\mathcal{C}$,
\[ detA_S\mid_{R=R_r}=-\frac{\mathfrak{p}_1\mathfrak{p}_2\gamma_S^4}{\alpha_S^4}(-R_r \alpha_S^2+s_3^2\pi^2 Q \gamma_S^2+\gamma_S^6)<0. \]
\end{itemize}
This ensures that \eqref{CM3} can always be solved when $S\notin\mathcal{C}$.
\subsection{Proof of Theorem \protect\ref{rollthm}.}
Now let there be a single critical index $J=(j,0,1)$  in $\mathcal{C}$. A simple calculation yields that
\[
G(\psi_J,\psi_J,e_{2j,0,2})=0.
\]
Thus 
\begin{equation}
\Phi(x)= x^2 \sum_{S\in\mathcal{S}_{roll}}\Phi_Se_{S}+o(2),
\end{equation}
where
\[\mathcal{S}_{roll}=\{(0,0,2),(2j,0,0)\}.\]
To compute the coefficients $\Phi_S$, we multiply \eqref{cmrealformula} by $e_S$ to get
\begin{equation}\label{RollCM}
\Phi_S=\frac{-G(\psi_J,\psi_J,e_S)}{\langle e_S,\mathcal{L}^{\ast}e_S\rangle}, \qquad S\in\mathcal{S}_{roll}.
\end{equation}
Thus we can write \eqref{red1} as:
\begin{equation}\label{realreducedfinal}
 \frac{dx}{dt}=\beta(R)x+b x^3+o(3),
\end{equation}
where
\begin{equation} \label{bJ} 
b=\sum_{S\in\mathcal{S}_{roll}}\frac{\Phi_S G_s(e_S,\psi_J,\psi_J^{\ast})}{\langle\psi_J,\psi_J^{\ast}\rangle}.
\end{equation}
Using \eqref{CM1} and \eqref{CM2}, we find that
\begin{equation}\label{rollcomp1}
\Phi_{0,0,2}=-\frac{\gamma^2}{8\pi}, \qquad
\Phi_{2j,0,0}=\frac{\pi^2\gamma^2}{4 \mathfrak{p}_2 \alpha^2}.
\end{equation}
Also directly computing, we have,
\begin{equation}\label{rollcomp2}
\begin{aligned}
& G_s(\psi_J,e_{0,0,2},\psi_J^{\ast})=\frac{L_1 L_2}{4} \mathfrak{p}_1\mathfrak{p}_2 \pi\gamma^2 R, \\
& G_s(\psi_J,e_{2j,0,0},\psi_J^{\ast})=\frac{L_1 L_2}{4 \alpha^2} \mathfrak{p}_1 \pi^2\gamma^2(\pi^2-\alpha^2) Q.
\end{aligned}
\end{equation}

\begin{equation}\label{psirpsir}
\langle\psi_J,\psi_J^{\ast}\rangle=\frac{L_1 L_2\gamma^2}{4 \alpha^2} (\mathfrak{p}_2\gamma^4(1+\mathfrak{p}_1)+\mathfrak{p}_1(\mathfrak{p}_2-1)\pi^2 Q) .
\end{equation}

When $\mathfrak{p}_2\geq 1$, or if $\mathfrak{p}_2<1$ and $Q<Q_0$, (see \eqref{Q0}) we see that $\langle\psi_J,\psi_J^{\ast}\rangle>0.$ Thus $\langle\psi_J,\psi_J^{\ast}\rangle$ has no effect on the transition.
Next plugging \eqref{rollcomp1}-\eqref{rollcomp2} into \eqref{bJ}, we find that at the critical Rayleigh number,
\begin{equation} \label{signofcubicroll}
 b=\frac{L_1 L_2\gamma^4}{32\mathfrak{p}_2 \alpha^4\langle \psi_J,\psi_J^{\ast}\rangle}( 2\pi^4 Q(\pi^2-\alpha_J^2)-\mathfrak{p}_2^2\alpha_J^4R_r).
\end{equation}

The transition in this case depends on $b$.  If $b<0$ then the transition is Type-I and the roll pattern with index $J$ is stable after the transition. On the other hand if $b>0$ then the transition is Type-II and the roll pattern with index $J$ is unstable after  the transition. Note that $b$ in \eqref{signofcubicroll} and $b$ in \eqref{a,b} have the same sign. That finishes the proof of  the first assertion.

Now we will prove the second assertion. For this let $J=(j,k,1)$ with $j\neq0$, $k\neq0$ and consider 
\[\mathcal{S}_{rec}= \{(0,0,2), (2j,2k,0), (2j,0,0), (0,2k,0), (0,2k,2), (2j,0,2)\}. \] 

It is easy to check that the center manifold function is:
\begin{equation} \label{rechex}
\begin{aligned}
\Phi=x^2\sum_{S\in\mathcal{S}_{rec}} \Phi_{S}e_{S}+O(3).
\end{aligned}
\end{equation}
As in the proof of the first assertion, the crucial parameter to find is the coefficient of the cubic term of the reduced equation which is given in a similar fashion as $b$ in \eqref{bJ}:
\begin{equation}\label{aJ}
a=\sum_{S\in\mathcal{S}_{rec}}\frac{\Phi_{S} G_s(e_{S},\psi_J,\psi_J^{\ast})}{\langle\psi_J,\psi_J^{\ast}\rangle}.
\end{equation}
By \eqref{CM1} and \eqref{CM2}, we can compute:
\begin{equation}\label{recres1}
\begin{aligned}
& \Phi_{S}=\frac{\pi^2\gamma^2}{8\mathfrak{p}_2\alpha^2}, \quad S=(2j,2k,0),(2j,0,0),(0,2k,0),\\
& \Phi_{0,0,2}=-\frac{\gamma^2}{16\pi}.
\end{aligned}
\end{equation}
By \eqref{CM3} we can compute for $S=(2j,0,2),(0,2k,2)$,
\begin{equation}\label{recres2}
\left[ \begin{array}{c} \Phi_S^1\\\Phi_S^2\\\Phi_S^3\end{array}\right]=-\frac{ \pi \gamma^2(4\alpha^2-\alpha_S^2)(R+\eta_S)}{16\alpha^2\gamma_S^2(R-R_S)}\left[\begin{array}{c}  
\gamma_S^2\\ 
\mathfrak{p}_1\mathfrak{p}_2 \alpha^2(R_S+\eta_S)(R+\eta_S)^{-1}\\ 
2\pi
\end{array}\right],
\end{equation}
where $\eta$ and $R_S$ are defined in \eqref{kappa}.
To compute $a$, we also need the following:
\begin{equation}\label{recres3}
\begin{aligned}
& G_s(\psi_J,e_{2j,2k,0},\psi_J^{\ast})=\frac{L_1 L_2}{16\alpha^2}\mathfrak{p}_1\pi^2 Q\gamma^2(\pi^2-\alpha^2),\\
& G_s(\psi_J,e_S,\psi_J^{\ast})=\frac{L_1 L_2}{8\alpha^4}\mathfrak{p}_1\pi^2 Q\gamma^2((-\alpha^2+\frac{1}{2}\alpha_S^2)\pi^2-\alpha^4) \quad S=(2j,0,0),(0,2k,0),\\
& G_s(\psi_J,e_{0,0,2},\psi_J^{\ast})=\frac{L_1 L_2}{8}\mathfrak{p}_1\mathfrak{p}_2 \pi R\gamma^2,
\end{aligned}
\end{equation}
For $S=(2j,0,2),(0,2k,2)$, we have:
\begin{equation}\label{recres4}
\left[ \begin{array}{c} 
G_s(\psi_J,e_S^1,\psi_J^{\ast})\\
G_s(\psi_J,e_S^2,\psi_J^{\ast})\\
G_s(\psi_J,e_S^3,\psi_J^{\ast})
\end{array}\right]
=\frac{L_1 L_2  \pi \gamma^2(4\alpha^2-\alpha_S^2)}{64\alpha^4}
\left[
\begin{array}{c}  
\mathfrak{p}_2\gamma^4-\mathfrak{p}_1\pi^2 Q \\  
\mathfrak{p}_1\mathfrak{p}_2 R \alpha^2 \\
-2\mathfrak{p}_1\pi Q\gamma^2
\end{array}\right].
\end{equation}
Now using \eqref{recres1}, \eqref{recres2} and after some simplification, we find that $a$  in \eqref{aJ} has the same sign as $a$ in \eqref{a,b}. The rest of the proof is same as in the first assertion. That proves the second assertion.

Finally, notice that by using \eqref{polynom}, the first critical eigenvalue can be written as $\beta(R)=c (R-R_c)+o(R-R_c)$ where $c\neq0$ is a constant. That proves the third assertion and the  proof of the Theorem \ref{rollthm}  is finished.

\subsection{Proof of Corollary \ref{rollcorollary}}
For the proof we will need the following lemma.
\begin{lemma}\label{minalpha} Let $\alpha_{J_r}$ be the critical wave number minimizing \eqref{Rr}. Then
$$\alpha_{J_r}\geq \frac{\pi}{2^{1/3}(2^{2/3}+1)^{1/2}},$$ for all $L_1,L_2,Q$.
\end{lemma}
\begin{proof}
Fix $L_1,L_2>0$. Let 
$$ R_1(x,Q)=\frac{ (\pi^2+x^2)}{x^2}((\pi^2+x^2)^2+Q \pi^2).$$ 
Due to convexity of $R_1$ with respect to $x$, for fixed $Q$, $R_1$ has a unique global minimizer $x_r$ on $x>0$. Now, there are two numbers $\alpha_1,\alpha_2$ of the form 
\[\sqrt{j^2 \pi^2/L_1^2+k^2\pi^2/L_2^2}, \quad j,k\in\mathbb{Z}\] 
such that $\alpha_{J_r}=\alpha_1$ or $\alpha_{J_r}=\alpha_2$ where $\alpha_1\leq x_r \leq\alpha_2$ and there are no other numbers of this form between $\alpha_1,\alpha_2$.

Since $x_r$ is an increasing function of $Q$, $\alpha_{J_r}$ is also an increasing function of $Q$. So it suffices to assume that $Q=0$. In this case $x_r\approx 2.22$.
We write:
$$ \alpha_{J_r}^2=\frac{j_r^2 \pi^2}{L_1^2}+\frac{k_r^2 \pi^2}{L_2^2}.$$
Now fix $L_1$. First, we will consider that $L_2/L_1$ is sufficiently small so that $k_r=0$.  Define 
$$L(m)=((m+1)m)^{1/3}((m+1)^{2/3}+m^{2/3})^{1/2}, \quad m\in\mathbb{Z},m\geq 0.$$
We claim that if $L(m-1)<L\leq L(m)$ for $m\geq1$ then $j_r=m$, i.e. $\alpha_{J_r}=\frac{\pi m}{L}$. In fact, $j_r\in \mathbb{Z}$ is an increasing function of $L$ (when $k=0$) and it can be easily checked that
$$R_1\left(\frac{m \pi}{L(m+1)}\right)=R_1\left(\frac{(m+1)\pi}{L(m+1)}\right).$$
Thus, if $L(m-1)<L<L(m)$ then $$\alpha_{J_r}= \frac{m \pi}{L}\geq \frac{m \pi}{L(m)}\geq \frac{\pi}{L(1)}\geq\frac{\pi}{2^{1/3}(2^{2/3}+1)^{1/2}}\geq 1.55.$$

Now if $L_2/L_1$ is not small, then $k_r$ is not necessarily zero and $\alpha_{J_r}$ will be closer to $x_r$. Thus the proof of the lemma is finished.\end{proof}

Now, to prove the Corollary \ref{rollcorollary}, consider $b$ defined by \eqref{a,b}. We notice that the condition $b>0$ dissects into three conditions:
\begin{equation}\label{rollcond}
\alpha_{J}<\pi, \qquad \mathfrak{p}_2<\frac{\pi\left(2(\pi^2-\alpha_{J}^2)\right)^{1/2}}{\alpha_{J}\gamma_{J}},\qquad Q \pi^2>\frac{\mathfrak{p}_2^2\alpha_{J}^2\gamma_{J}^6}{2\pi^2(\pi^2-\alpha_{J}^2)-\mathfrak{p}_2^2\alpha_{J}^2\gamma_{J}^2}.
\end{equation}
To make use of \eqref{rollcond}, it is important to get bounds on the wave number $\alpha_J$. One can see that there is no upper bound on the wave number by letting $L_1,L_2\rightarrow0$.  However, Lemma \ref{minalpha} shows that the minimum wave number can not go below $\approx1.55$. Moreover by definition of $R_r$, see \eqref{Rr}, $\alpha_{J}$ will increase as $Q$ increases. And using a similar argument as in the proof of Lemma \ref{minalpha} we find that $\alpha_{J}>\pi$ for $Q>306$. Thus the transition is Type-I regardless of the container size when $Q>306$.
Next, using the fact that $\alpha_J>1.55$ and the second condition in \eqref{rollcond}, one sees that the transition is always Type-I if $\mathfrak{p}_2>2.24$ regardless of the wave number.

When at least one of the length scales $L_1,L_2$ is much larger than 1, then the wave number will be close to the absolute minimizer of \eqref{Rr}. The relation between $Q$ and the critical wave number is then found to be:
\begin{equation}\label{Qalpharoll}
Q \pi^4=-\pi^6+3\pi^2\alpha_J^4+2\alpha_J^6.
\end{equation}
The minimum $\alpha_J$ can be found to be $\alpha_J=\pi/\sqrt{2}$ by setting $Q=0$ above. Then the second condition in \eqref{rollcond} implies that the transition is Type-I if $\mathfrak{p}_2>2/\sqrt{3}>1.15$. Finally by \eqref{Qalpharoll} one finds that for $Q>4\pi^2$, $\alpha_J>\pi$. 

\subsection{Proof of Theorem \ref{hexthm}}
Let $I=(j,k,1)$ and $J=(0,2k,1)$ be the critical indices with $j,k\geq 1$.

Consider $S_1=(0,0,2)$, $S_2=(2j,2k,0)$, $S_3=(2j,0,0)$, $S_4=(0,2k,0)$,  $S_5=(0,4k,0)$, $S_6=(j,k,2)$, $S_7=(j,3k,2)$, $S_8=(j,3k,0)$, $S_9=(j,k,0)$, $S_{10}=(0,2k,2)$, $S_{11}=(2j,0,2)$. In this case the center manifold function is:
\begin{equation} \label{hex1}
\begin{aligned}
\Phi=x_I^2(\Phi_{S_1}^I e_{S_1}+\sum_{\substack{i=2,3,\\4,10,11}} \Phi_{S_i}e_{S_i})+x_Ix_J(\sum_{\substack{i=6,\\7,8,9}} \Phi_{S_i} e_{S_i})+x_J^2(\Phi_{S_1}^J e_{S_1}+\Phi_{S_5} e_{S_5}).
\end{aligned}
\end{equation}
After letting $x=x_I, y=x_J$, the reduced equations \eqref{red1}  become:
\begin{equation}\label{hexred1}
\begin{aligned}
&\frac{dx}{dt}=\beta x+x(a_1 x^2+a_2 y^2)+o(3),\\
&\frac{dy}{dt}=\beta y+y(b_1 x^2+b_2y^2)+o(3).
\end{aligned}
\end{equation}
The coefficients in \eqref{hexred1} can be found as:
\begin{equation}\label{hexcoefs}
\begin{aligned}
& a_1=\frac{1}{\langle \psi_I,\psi_I^{\ast}\rangle}\left(\Phi_{S_1}^I G_s(\psi_I,e_{S_1},\psi_I^{\ast})+\sum_{i=2,3,4,10,11} \Phi_{S_i}G_s(\psi_I,e_{S_i},\psi_I^{\ast})\right), \\
& a_2=\frac{1}{\langle \psi_I,\psi_I^{\ast}\rangle}\left(\Phi_{S_1}^J G_s(\psi_I,e_{S_1},\psi_I^{\ast})+\sum_{i=6,7,8,9} \Phi_{S_i}G_s(\psi_J,e_{S_i},\psi_I^{\ast})\right), \\
& b_1=\frac{1}{\langle \psi_J,\psi_J^{\ast}\rangle}\left(\Phi_{S_1}^I G_s(\psi_J,e_{S_1},\psi_J^{\ast})+\sum_{i=6,7,8,9} \Phi_{S_i}G_s(\psi_I,e_{S_i},\psi_J^{\ast})\right), \\
& b_2=\frac{1}{\langle \psi_J,\psi_J^{\ast}\rangle}\left(\Phi_{S_1}^J G_s(\psi_J,e_{S_1},\psi_J^{\ast})+\Phi_{S_5}G_s(\psi_J,e_{S_5},\psi_J^{\ast})\right).
\end{aligned}
\end{equation}

We can find the following relations:
\begin{equation}
\begin{aligned}
& G_s(\psi_J,e_{S_1},\psi_J^{\ast})=2 G_s(\psi_I,e_{S_1},\psi_I^{\ast}), \quad G_s(\psi_J,e_{S_5},\psi_J^{\ast})=4 G_s(\psi_I,e_{S_2},\psi_I^{\ast}), \\
& G_s(\psi_I,\psi_J,e_{S_6})=2 G(\psi_I,\psi_I,e_{S_{10}}), \quad G_s(\psi_I,\psi_J,e_{S_7})=2 G(\psi_I,\psi_I,\psi_{S_{11}}), \\
& G_s(\psi_I,e_{S_6},\psi_J^{\ast})= G_s(\psi_I,e_{S_{10}},\psi_I^{\ast}), \quad G_s(\psi_I,e_{S_7},\psi_J^{\ast})=  G_s(\psi_I,e_{S_{11}},\psi_I^{\ast}), \\
& G_s(\psi_J,e_{S_i},\psi_I^{\ast})=G_s(\psi_I,e_{S_i},\psi_J^{\ast}), \quad i=6,7,8,9.
\end{aligned}
\end{equation}

\begin{equation}
\begin{aligned}
& \Phi_{S_1}^J=2 \Phi_{S_1}^I, \Phi_{S_5}=2 \Phi_{S_2},  \\
& 4\sum_{i=3,4}\Phi_{S_i} G_s(\psi_I,e_{S_i},\psi_I^{\ast})=\sum_{i=8,9}\Phi_{S_i} G_s(\psi_J,e_{S_i},\psi_I^{\ast}), \\
& 4\sum_{i=6,7}\Phi_{S_i} G_s(\psi_I,e_{S_i},\psi_I^{\ast})=\sum_{i=10,11}\Phi_{S_i} G_s(\psi_J,e_{S_i},\psi_I^{\ast}).
\end{aligned}
\end{equation}
Using these relations, we can write \eqref{hex1} as:
\begin{equation}\label{hexfinal}
\begin{aligned}
& \frac{dx}{dt}=\beta x+x\left(ax^2+2(2a-b)y^2\right)+o(3), \\
& \frac{dy}{dt}=\beta y+y((2a-b)x^2+2b y^2)+o(3).
\end{aligned}
\end{equation}
Here $a=a_1$ and $2b=b_2$ in \eqref{hexcoefs}. After simplification of same positive terms appearing, we find that $a$ and $b$ are given by \eqref{a,b}.

The equations \eqref{hexfinal} have eight straight line orbits in total which are on the lines: x-axis, y-axis and $x=\pm2y$. The eight bifurcated solutions of \eqref{hexfinal} are given by:

\begin{equation}\label{hexbifsols}
\begin{aligned}
& Rec_{\pm}=\pm\sqrt{-\beta/2b}\,(0,1), \quad Roll_{\pm}=\pm\sqrt{-\beta/a}\,(1,0), \\
& Hex_{\pm}^{i}=(-1)^i\sqrt{-\beta/2(4a-b)}\,(2, \pm1), \quad i=1,2.
\end{aligned}
\end{equation}

Moreover the Jacobian matrix of the vector field \eqref{hexfinal} has the following eigenvalues at these bifurcated points:
\begin{equation}\label{regularity}
\begin{array}{cc}
 \lambda^1_{Roll_{\pm}}=-\frac{4\beta(a-b)}{b}, &  \lambda^2_{Roll_{\pm}}=-2\beta, \\
  \lambda^1_{Rec_{\pm}}=-\frac{4\beta(a-b)}{2a}, &  \lambda^2_{Rect_{\pm}}=-2\beta, \\
   \lambda^1_{Hex^{1,2}_{\pm}}=\frac{2\beta(a-b)}{4a-b}, &  \lambda^2_{Hex^{1,2}_{\pm}}=-2\beta. \\
\end{array}
\end{equation}
Using \eqref{hexfinal}-\eqref{regularity}, we see that the transition depends on the sign of $a$, $b$, $4a-b$ and $a-b$. These 4 lines separate the $a,b$ plane into eight regions. In each region, it is possible to find which solution is bifurcated on $\beta>0$, $\beta<0$ and the stability of these solutions. The analysis is given from Figure \ref{regions2} to Figure \ref{VIII}. These eight cases exhaust all possible transition scenarios. The proof is finished.

\subsection{Proof of Corollary \ref{hexcor}}
First notice that 
\[
\begin{aligned}
& \gamma_{2j ,0,2}^2=4\alpha_{j, 0}^2+4\pi^2<4\alpha_{j,k}^2+4\pi^2=4\gamma_{j,k,1}^2 \\
& \gamma_{0,2k,2}^2=4\alpha_{0,k}^2+4\pi^2<4\alpha_{j,k}^2+4\pi^2=4\gamma_{j,k,1}^2 \\
\end{aligned}
\]
Thus 
\[
\begin{aligned}
& \mathfrak{p}_2
\geq 8
\geq \frac{\gamma_{2j,0,2}^2}{\gamma_J^2}+4
\geq \frac{\gamma_{2j,0,2}^2}{\gamma_J^2}(1+4\frac{\gamma_{2j,0,2}^2}{\gamma_J^2})
\geq \frac{\gamma_{2j,0,2}^2}{\gamma_J^2}(1+4\frac{\gamma_{2j,0,2}^2}{\gamma_J^2}) \frac{\pi^2 Q}{\gamma_J^4+\pi^2 Q}\\
& \mathfrak{p}_1\pi^2 Q(1+4\frac{\gamma_J^2}{\gamma_{2j,0,2}^2})
\leq\mathfrak{p}_1\mathfrak{p}_2 \frac{\gamma_J^2}{\gamma_{2j,0,2}^2}(\gamma_J^4+\pi^2 Q)\\
& \mathfrak{p}_1\pi^2 Q(1+4\frac{\gamma_J^2}{\gamma_{2j,0,2}^2})-\mathfrak{p}_2\gamma_J^4
\leq\mathfrak{p}_1\mathfrak{p}_2 \frac{\gamma_J^2}{\gamma_{2j,0,2}^2}(\gamma_J^4+\pi^2 Q)
\end{aligned}
\]
Noticing that $R_R<R_{2j,0,2}$ we see that $\kappa_{2j,0}<0$ with $\kappa$ defined in \eqref{kappa}. The same argument shows also that $\kappa_{0,2k}<0$. By Lemma \ref{minalpha}, $\pi^2-5\alpha_J^2<0$. That shows that $a<0$ where $a$ is defined in \eqref{a,b}. Now by Corollary \ref{rollcorollary}, $b<0$ when $\mathfrak{p}_2>2.24$. We can write $a/b$ as:
\[
\frac{a}{b}=1+\frac{\pi^4 Q(\pi^2-5\alpha_J^2)-2\pi^4Q(\pi^2-\alpha_J^2)-\mathfrak{p}^2\alpha_J^2 R_R+\kappa_{2j,0}+\kappa_{0,2k}}{b}>1.
\]
Now, by Theorem \ref{hexthm}, the transition depends on $a/b$ and we see that the system is in Region $\mathcal{I}_2$ when $a/b>1, a<0,b<0$. That proves the corollary.
\subsection{Proof of Theorem \protect\ref{Complex MHD}}
Define $\beta \left( R\right) =\beta _{J_{c}}^{1}\left( R\right) =\lambda
\left( R\right) +i\rho \left( R\right) $ with $\lambda \left( R_{c}\right)
=0 $, $\rho \left( R_{c}\right) =\rho >0$. The corresponding eigenvectors
and conjugate eigenvectors are 
\begin{align*}
& \psi _{J_{c}} =\psi _{J_{c}}^{1}+i\psi _{J_{c}}^{2}\text{, \ \ }\psi
_{J_{c}}^{1}=Re\psi _{J_{c}}\text{, \ \ }\psi _{J_{c}}^{2}=Im\psi _{J_{c}},
\\
& \psi _{J_{c}}^{\ast } =\psi _{J_{c}}^{\ast 1}+i\psi _{J_{c}}^{\ast 2}\text{%
, \ \ }\psi _{J_{c}}^{\ast 1}=Re\psi _{J_{c}}^{\ast }\text{, \ \ }\psi
_{J_{c}}^{\ast 2}=Im\psi _{J_{c}}^{\ast }.
\end{align*}
The horizontal components of the critical eigenvectors can be captured as in \eqref{horizontals}. The amplitudes of the vertical velocity, the vertical magnetic field and the temperature of $\psi _{J_{c}}$, $\psi _{J_{c}}^{\ast }$ are given by: 
\begin{align}
\begin{split} \label{ev3}
& W_J =(\mathfrak{p}_2^{-1} \beta +\gamma _{J}^{2}) 
\alpha _{J}^{2}R,   \\
& \Theta_J =\omega_J(\beta)\gamma _{J}^{2},  \\
& H_{3,J} =\alpha _{J}^{2}l\pi R,  
\end{split}
\end{align}
\begin{align}
\begin{split} \label{ev*3} 
& W_{J}^{\ast} =\left( \overline{\beta }+\mathfrak{p}_2\gamma
_{J}^{2}\right) \mathfrak{p}_1^{-1}\alpha _{J}^{2},  \\
& \Theta_{J}^{\ast } =\mathfrak{p}_2\omega _{J}(\overline{\beta })\gamma _{J}^{2},  \\
& H_{3,J}^{\ast} =-l\pi Q\alpha _{J}^{2}. 
\end{split}
\end{align}
Here $\omega_J(\beta)=\left( \frac{\beta }{\mathfrak{p}_2}+\gamma
_{J}^{2}\right) \left( \frac{\beta }{\mathfrak{p}_1}+\gamma _{J}^{2}\right) +Q\left(
l\pi \right) ^{2}.$ Then
\begin{align*}
& L\psi _{J_{c}}^{1} =\lambda \psi _{J_{c}}^{1}-\rho \psi _{J_{c}}^{2},\quad
L\psi _{J_{c}}^{2}=\rho \psi _{J_{c}}^{1}+\lambda \psi _{J_{c}}^{2}, \\
& L\psi _{J_{c}}^{\ast 1} =\lambda \psi _{J_{c}}^{\ast 1}-\rho \psi
_{J_{c}}^{\ast 2},\quad L\psi _{J_{c}}^{\ast 2}=\rho \psi _{J_{c}}^{\ast
1}+\lambda \psi _{J_{c}}^{\ast 2},
\end{align*}
and
\begin{equation}
\left( \psi _{J_{c}}^{1},\psi _{J_{c}}^{\ast 1}\right) =\left( \psi
_{J_{c}}^{2},\psi _{J_{c}}^{\ast 2}\right) \text{, \quad }\left( \psi
_{J_{c}}^{1},\psi _{J_{c}}^{\ast 2}\right) =\left( \psi _{J_{c}}^{2},\psi
_{J_{c}}^{\ast 1}\right) .  \label{evrelt} 
\end{equation}

We define the eigenvectors
\begin{align}
& \Phi _{J_{c}}^{\ast 1}=\psi _{J_{c}}^{\ast 1}-C\psi _{J_{c}}^{\ast 2},\qquad
\Phi _{J_{c}}^{\ast 2}=C\psi _{J_{c}}^{\ast 1}+\psi _{J_{c}}^{\ast 2},\qquad
C=\frac{\left( \psi _{J_{c}}^{2},\psi _{J_{c}}^{\ast 1}\right) }{\left( \psi
_{J_{c}}^{2},\psi _{J_{c}}^{\ast 2}\right) }.  \label{phi*1}
\intertext{Then by (\ref{evrelt}),}
& \left( \psi _{J_{c}}^{1},\Phi _{J_{c}}^{\ast 2}\right) =\left( \psi
_{J_{c}}^{2},\Phi _{J_{c}}^{\ast 1}\right)=0 , \notag \\
& B :=\left( \psi _{J_{c}}^{1},\Phi _{J_{c}}^{\ast 1}\right) ^{-1}=\left(
\psi _{J_{c}}^{2},\Phi _{J_{c}}^{\ast 2}\right) ^{-1}\neq 0.\notag
\end{align}
Now we write $\psi =x\psi _{J_{c}}^{1}+y\psi _{J_{c}}^{2}+\Phi $ where $\Phi 
$ is the center manifold function, $x,y\in \mathbb{R}$. We can write the
reduced equations of (\ref{Eqn1}), (\ref{B.C.}) as 
\begin{equation}
\begin{aligned} \label{complex-reduced-1} 
& \frac{dx}{dt} =\lambda x+\rho y+B\left( G\left( \psi ,\psi \right) ,\Phi
_{J_{c}}^{\ast 1}\right) ,  \\
& \frac{dy}{dt} =-\rho x+\lambda y+B\left( G\left( \psi ,\psi \right) ,\Phi
_{J_{c}}^{\ast 2}\right) 
\end{aligned}
\end{equation}
We have the following approximation formula of the center manifold function, see 
\cite{ptd}. For $z=x\psi _{J_{c}}^{1}+y\psi _{J_{c}}^{2}\in E_{1}$, 
\begin{align*}
\left( \left( -\mathcal{L}_{R}\right) ^{2}+4\rho ^{2}\right) \left( -%
\mathcal{L}_{R}\right) \Phi \left( z,R\right)  =& \left( \left( -\mathcal{L}%
_{R}\right) ^{2}+4\rho ^{2}\right) P_{2}G\left( z,R\right) \\
& -2\rho ^{2}P_{2}G\left( z,R\right) +2\rho ^{2}P_{2}G\left( x\psi
_{J_{c}}^{2}-y\psi _{J_{c}}^{1},R\right) \\
& +\rho \left( -\mathcal{L}_{R}\right) \left( G\left( x\psi
_{J_{c}}^{1}+y\psi _{J_{c}}^{2},y\psi _{J_{c}}^{1}-x\psi
_{J_{c}}^{2},R\right) \right) \\
& +G\left( y\psi _{J_{c}}^{1}-x\psi _{J_{c}}^{2},x\psi _{J_{c}}^{1}+y\psi
_{J_{c}}^{2},R\right) +o\left( 2\right) ,
\end{align*}%
where%
\begin{equation*}
o\left( 2\right) =o\left( \left\Vert z\right\Vert ^{2}\right) +O\left(
\left\vert \lambda \left( R\right) \right\vert \left\Vert z\right\Vert
^{2}\right) .
\end{equation*}
Here $\mathcal{L}_{R}=L_{R}\mid _{E_{2}}$ for $H=E_{1}\oplus E_{2}$, $%
E_{1}=span\{\psi _{J_{c}}^{1},\psi _{J_{c}}^{2}\}$, $E_{2}=E_{1}^{\perp }$.
As in the real case:
\begin{equation}  \label{complex-nonlinear-result-2}
\left( G\left( \psi _{J_{c}}^{i},\psi _{J_{c}}^{j}\right) ,\psi _{J}^{\ast
}\right) =0 \text{   if $J\neq \left( 0,0,2\right) $ or $J\neq \left( 2j_{c},2k_{c},0\right) $}. 
\end{equation}
By (\ref{complex-nonlinear-result-2}) and
the above formula, we have the following approximation for the center manifold
\begin{equation} \label{complex-center-manifold}
\Phi \left( z,R\right) =\Phi _{1}\psi _{002}+\Phi _{2}\psi
_{2j_{c}2k_{c}0}+o\left( 2\right) ,
\end{equation}%
Here $\Phi_1$ and $\Phi_2$ are given by:
\begin{equation*}
\begin{aligned}
& \Phi _{1}=A_{1}x^{2}+A_{2}xy+A_{3}y^{2},  \\
& \Phi _{2}= A_{4}x^{2}+A_{5}xy+A_{6}y^{2} .
\end{aligned}
\end{equation*}
To compute $A_i$, we  list the necessary results 
\begin{align}
& g_{ij}^{1} :=\left( G\left( \psi _{J_{c}}^{i},\psi _{J_{c}}^{j}\right)
,\psi _{002}^{\ast }\right) ,  \label{complex-nonlinear-result-1}\\
& g_{11}^{1} =- ReL_{\beta} Re\omega_{\beta} , \qquad
  g_{12}^{1} =- ReL_{\beta} Im\omega_{\beta} , \notag \\
& g_{21}^{1} =- ImL_{\beta} Re\omega_{\beta} , \qquad
  g_{22}^{1} =- ImL_{\beta} Im\omega_{\beta} ; \notag \\
\intertext{ }
& g_{ij}^{2} :=\left( G\left( \psi _{J_{c}}^{i},\psi _{J_{c}}^{j}\right)
,\psi _{2j_{c}2k_{c}0}^{2\ast }\right) ,  \label{complex-nonlinear-result-1*}\\
& g_{11}^{2} =-4Q^{-1}\alpha_{J_{c}}^2R ReK_{\beta} ,  \qquad
  g_{12}^{2} =-2Q^{-1}\alpha_{J_{c}}^2R ImK_{\beta} ,  \notag \\
& g_{21}^{2} =g_{12}^{2},  \qquad
  g_{22}^{2} =0;  \notag \\
\intertext{ }
& c_{ij}^{1} :=\left( G\left( \psi _{J_{c}}^{i},\psi _{002}\right) ,\psi
_{J_{c}}^{\ast j}\right) ,  \label{complex-nonlinear-result-4}\\
& c_{11}^{1} = \mathfrak{p}_2 ReL_{\beta} Re\omega_{\overline\beta} ,  \qquad 
  c_{12}^{1} = \mathfrak{p}_2 ReL_{\beta} Im\omega_{\overline\beta} ,  \notag \\
& c_{21}^{1} = \mathfrak{p}_2 ImL_{\beta} Re\omega_{\overline\beta} ,  \qquad
  c_{22}^{1} = \mathfrak{p}_2 ImL_{\beta} Im\omega_{\overline\beta} ;  \notag
\intertext{ }
& c_{ij}^{2} :=\left( G\left(\delta_{\beta}{J_{c}}^{i},\psi _{2j_{c}2k_{c}0}\right)
,\psi _{J_{c}}^{\ast j}\right) ,  \label{complex-nonlinear-result-4*} \\
& c_{11}^{2} =4\alpha_{J_{c}}^2 Re K_{\beta} ,  \qquad
  c_{12}^{2} =2\alpha_{J_{c}}^2 Im K_{\overline{\beta}} , \notag \\
& c_{21}^{2} =2\alpha_{J_{c}}^2 ImK_{\beta} ,  \qquad
  c_{22}^{2} =0;  \notag
\intertext{ }
& d_{ij} :=\left( G\left( \psi _{2j_{r}2k_{r}0},\psi _{J_{c}}^{i}\right)
,\psi _{J_{c}}^{\ast j}\right) ,  \label{complex-nonlinear-result-5*} \\
& d_{11} =2 ( \alpha_{J_{c}}^{2}-\pi ^{2}) Re K_{\beta} ,  \qquad 
  d_{12} =( \alpha_{J_{c}}^{2}-\pi ^{2}) Im K_{\overline{\beta}},  \notag \\
& d_{21} =( \alpha_{J_{c}}^{2}-\pi ^{2}) Im K_{\beta} ,  \qquad
  d_{22} =0.  \notag
\intertext{Here}
& \omega _{\beta}=\left( \frac{\beta }{\mathfrak{p}_2}%
+\gamma _{J_{c}}^{2}\right) \left( \frac{\beta }{\mathfrak{p}_1}+\gamma
_{J_{c}}^{2}\right) +Q\pi ^{2}, \notag \\
& \delta _{\beta}= \beta +\mathfrak{p}_2\gamma _{J_{c}}^{2},  \notag \\
& K_{\beta}=-\frac{L_{1}L_{2}}{8}QR\alpha _{J_{c}}^{2}\pi ^{2} \mathfrak{p}_2^{-1} \delta_{\beta},  \notag \\
& L_{\beta}=\frac{L_{1}L_{2}}{2}\gamma _{J_{c}}^{2}\alpha_{J_{c}}^{2}R\pi \mathfrak{p}_2^{-1} \delta_{\beta}. \notag
\end{align}
Using these results, the coefficients of the center manifold function can be computed as follows:
\begin{align*}
& A_{1} =\frac{\left( g_{11}^{1}\left( 16\pi ^{4}+2\rho ^{2}\right)
+g_{22}^{1}2\rho ^{2}-g_{12}^{1}\rho 4\pi ^{2}-g_{21}^{1}\right) }{4\pi
^{2}\left( 16\pi ^{4}+4\rho ^{2}\right) }, \\
& A_{2} =\frac{\left( \left( g_{12}^{1}+g_{21}^{1}\right) 16\pi ^{4}+\left(
g_{11}^{1}-g_{22}^{1}\right) \left( \rho 4\pi ^{2}+1\right) \right) }{4\pi
^{2}\left( 16\pi ^{4}+4\rho ^{2}\right) }, \\
& A_{3} =\frac{\left( g_{22}^{1}\left( 16\pi ^{4}+2\rho ^{2}\right)
+g_{11}^{1}2\rho ^{2}+g_{21}^{1}\rho 4\pi ^{2}+g_{12}^{1}\right) }{4\pi
^{2}\left( 16\pi ^{4}+4\rho ^{2}\right) }, \\
& A_{4} =\frac{\left( 16\mathfrak{p}_2^{2}\alpha _{J_{c}}^{4}+2\rho
^{2}\right) g_{11}^{2}-\left( 4\rho \mathfrak{p}_2\alpha
_{J_{c}}^{2}+1\right) g_{12}^{2}}{4\mathfrak{p}_2\alpha _{J_{c}}^{2}\left( 16%
\mathfrak{p}_2^{2}\alpha _{J_{c}}^{4}+4\rho ^{2}\right) }, \\
& A_{5} =\frac{\left( 4\rho \mathfrak{p}_2\alpha _{J_{c}}^{2}+1\right)
g_{11}^{2}+32\mathfrak{p}_2^{2}\alpha _{J_{c}}^{4}g_{12}^{2}}{4\mathfrak{p}_2%
\alpha _{J_{c}}^{2}\left( 16\mathfrak{p}_2^{2}\alpha _{J_{c}}^{4}+4\rho
^{2}\right) }, \\
& A_{6} =\frac{2\rho ^{2}g_{11}^{2}+\left( 4\rho \mathfrak{p}_2\alpha
_{J_{c}}^{2}+1\right) g_{12}^{2}}{4\mathfrak{p}_2\alpha _{J_{c}}^{2}\left( 16%
\mathfrak{p}_2^{2}\alpha _{J_{c}}^{4}+4\rho ^{2}\right) }.
\end{align*}

Plugging\eqref{complex-center-manifold} into (\ref{complex-reduced-1}), we obtain the following ODE:
\begin{equation}
\begin{aligned} &\frac{dx}{dt} =\lambda x+\rho
y+a_{30}^{1}x^{3}+a_{21}^{1}x^{2}y+a_{12}^{1}xy^{2}+a_{03}^{1}y^{3}+o(3) ,
\\ &\frac{dy}{dt} =-\rho x+\lambda
y+a_{30}^{2}x^{3}+a_{21}^{2}x^{2}y+a_{12}^{2}xy^{2}+a_{03}^{2}y^{3}+o(3).
\end{aligned}  \label{complex-reduced-2}
\end{equation}%
Using the approximation (\ref{complex-center-manifold}), we get
\begin{equation} \label{comp-non-1}
\begin{aligned}
\left( G\left( \psi ,\psi \right) ,\Phi _{J_{c}}^{\ast 1}\right) =\, &x\Phi
_{1}\left( c_{11}^{1}-Cc_{12}^{1}\right) +x\Phi _{2}\left(
c_{11}^{2}-Cc_{12}^{2}\right)+ \\
&y\Phi _{1}\left( c_{21}^{1}-Cc_{22}^{1}\right) +y\Phi _{2}\left(
c_{21}^{2}-Cc_{22}^{2}\right)+ \\
&x\Phi _{2}\left( d_{11}-Cd_{12}\right) +y\Phi _{2}\left(
d_{21}-Cd_{22}\right)+o(3) ,
\end{aligned}
\end{equation}
\begin{equation} \label{comp-non-2}
\begin{aligned} 
\left( G\left( \psi ,\psi \right) ,\Phi _{J_{c}}^{\ast 2}\right) =\, &x\Phi
_{1}\left( Cc_{11}^{1}+c_{12}^{1}\right) +x\Phi _{2}\left(
Cc_{11}^{2}+c_{12}^{2}\right)+ \\
&y\Phi _{1}\left( Cc_{21}^{1}+c_{22}^{1}\right) +y\Phi _{2}\left(
Cc_{21}^{2}+c_{22}^{2}\right)+ \\
&x\Phi _{2}\left( Cd_{11}+d_{12}\right) +y\Phi _{2}\left(
Cd_{21}+d_{22}\right)+o(3) .
\end{aligned}
\end{equation}

Now we give the coefficients of the reduced equation \eqref{complex-reduced-2}. Using (\ref{complex-center-manifold}), (\ref{comp-non-1}) and (\ref{comp-non-2}), the coefficients are computed by plugging  (\ref{comp-non-1}), (\ref{comp-non-2}) into (\ref{complex-reduced-1}) and are given by:
\begin{align*}
& a_{30}^{1} =B( A_{1}( c_{11}^{1}-Cc_{12}^{1}) +A_{4} X ), \\
& a_{21}^{1} = B ( A_{1}( c_{21}^{1}-Cc_{22}^{1}) +A_{2}(c_{11}^{1}-Cc_{12}^{1}) + A_{4}( c_{21}^{2}+d_{21})  +A_{5} X ) ,\\
& a_{12}^{1} = B( A_{2}( c_{21}^{1}-Cc_{22}^{1}) +A_{3}(c_{11}^{1}-Cc_{12}^{1}) +A_{5}( c_{21}^{2}+d_{21}) +A_{6} X ) , \\
& a_{03}^{1} =B( A_{3}( c_{21}^{1}-Cc_{22}^{1}) +A_{6}(c_{21}^{2}+d_{21}) ) , \\
& a_{30}^{2} =B( A_{1}( Cc_{11}^{1}+c_{12}^{1}) +A_{4} Y, \\
& a_{21}^{2} =B( A_{1}( Cc_{21}^{1}+c_{22}^{1}) +A_{2}(Cc_{11}^{1}+c_{12}^{1}) +A_{4}C( c_{21}^{2}+d_{21})+A_{5} Y , \\
& a_{12}^{2} =B( A_{2}( Cc_{21}^{1}+c_{22}^{1}) +A_{3}(Cc_{11}^{1}+c_{12}^{1}) +A_{5}C( c_{21}^{2}+d_{21})+A_{6} Y , \\
& a_{03}^{2} =B( A_{3}( Cc_{21}^{1}+c_{22}^{1})+A_{6}C( c_{21}^{2}+d_{21}) ), \\
& X= c_{11}^{2}+d_{11}-C( c_{12}^{2}+d_{12}),  \\
& Y= c_{12}^{2}+d_{12}+C( c_{11}^{2}+d_{11}). 
\end{align*}

The transition of (\ref{Eqn1})-(\ref{B.C.}) is determined by the sign of the
following number at $R=R_{c}$; see \cite{ptd},%
\begin{equation}
\frac{3\pi }{4}\left( a_{30}^{1}+a_{03}^{2}\right) +\frac{\pi }{4}\left(
a_{12}^{1}+a_{21}^{2}\right) .  \label{firstb}
\end{equation}%
which has the same sign as $b$ defined by:
\begin{equation}  \label{b1}
b=\frac{D_{1}+D_{2}}{\pi ^{2}\left( 16\pi ^{4}+4\rho ^{2}\right) }+\frac{%
Q\pi \left( -3\alpha _{J_{c}}^{2}+\pi ^{2}\right) \pi R_{c}}{2\mathfrak{p}_2%
\gamma _{J_{c}}^{2}\left( 16\mathfrak{p}_2^{2}\alpha _{J_{c}}^{4}+4\rho
^{2}\right) }D_{3}.
\end{equation}
Here
\begin{align*}
& D_{1}=2\mathfrak{p}_2\left( 3\gamma _{J_{c}}^{2}A_{1}+\rho A_{2}+\gamma
_{J_{c}}^{2}A_{3}\right) \left( E_{1}\psi _{11}+E_{2}\psi _{21}\right) , \\
& D_{2}=2\mathfrak{p}_2\left( \rho A_{1}+\gamma _{J_{c}}^{2}A_{2}+3\rho
A_{3}\right) \left( E_{1}\psi _{21}-E_{2}\psi _{11}\right) , \\
& D_{3}=A_{4} (3\psi _{11}+2\psi _{21}\rho \mathfrak{p}_2^{-1}\gamma
_{J_{c}}^{-2}) +A_{5}\psi _{21}+A_{6} (\psi _{11}+2 \psi _{21} \rho \mathfrak{p}_1^{-1}\gamma _{J_{c}}^{-2}) ,\\
& E_{1} =( \mathfrak{p}_2+\mathfrak{p}_1) ( \gamma _{J_{c}}^{4}\mathfrak{p}_1^{-1}+Q\pi ^{2}(\mathfrak{p}_1+1)^{-1}) , \\
& E_{2} =( \mathfrak{p}_2\mathfrak{p}_1) ^{-1} (\mathfrak{p}_2+\mathfrak{p}_1) \rho \gamma _{J_{c}}^{2},
\end{align*}
\begin{align*}
& A_{1} =-\left( \left( 16\pi ^{4}+2\rho ^{2}\right) \gamma
_{J_{c}}^{2}E_{1}+2\mathfrak{p}_2^{-1}\rho E_{2}\rho ^{2}-4\gamma
_{J_{c}}^{2}E_{2}\rho \pi ^{2}-E_{1}\mathfrak{p}_2^{-1}\rho \right) , \\
& A_{2} =-\left( 16\pi ^{4}\left( E_{2}\gamma _{J_{c}}^{2}+\mathfrak{p}_2%
^{-1}\rho E_{1}\right) +\left( \gamma _{J_{c}}^{2}E_{1}-\mathfrak{p}_2%
^{-1}\rho E_{2}\right) \left( 4\rho \pi ^{2}+1\right) \right) , \\
& A_{3} =-\left( \mathfrak{p}_2^{-1}\rho E_{2}\left( 16\pi ^{4}+2\rho^{2}\right) +2\rho ^{2}\gamma _{J_{c}}^{2}E_{1}+4\pi ^{2}\mathfrak{p}_2^{-1}
\rho ^{2}E_{1}+\gamma _{J_{c}}^{2}E_{2}\right) ,\\
& A_{4} =2\gamma _{J_{c}}^{2}\left( 16\mathfrak{p}_2^{2}\alpha_{J_{c}}^{4}+2\rho ^{2}\right) -\left( 4\rho \mathfrak{p}_2\alpha_{J_{c}}^{2}+1\right) \rho , \\
& A_{5} =2\gamma _{J_{c}}^{2}\left( 4\rho \mathfrak{p}_2\alpha_{J_{c}}^{2}+1\right) +32\mathfrak{p}_2^{2}\rho , \\
& A_{6} =4\rho ^{2}\gamma _{J_{c}}^{2}+\left( 4\rho \mathfrak{p}_2\alpha_{J_{c}}^{2}+1\right) \rho ,\\
& \psi _{11} =-( ( \mathfrak{p}_2\mathfrak{p}_1) ^{-1}\rho ^{2}\alpha_{J_{c}}^{2}R_{c}+\mathfrak{p}_2E_{2}^{2}\gamma _{c}^{2}) , \\
& \psi _{21} =( \mathfrak{p}_1^{-1}\rho \alpha _{J_{c}}^{2}R_{c}+\mathfrak{p}_2 E_{1}E_{2}) \gamma _{J_c}^{2}.
\end{align*}%

Finally from \eqref{Rc}, we notice that
\begin{equation}
\gamma_{J_{c}}^{2}=O\left( Q^{1/3}\right) \text{ as }Q\rightarrow \infty .
\end{equation}
Using this, we see that, as $Q\rightarrow \infty $,%
\begin{equation}\label{Qinfty}
\gamma _{J_{c}}^{2}\rightarrow cQ^{1/3},\rho ^{2}\rightarrow \frac{\mathfrak{p}_1%
\mathfrak{p}_2\left( 1-\mathfrak{p}_2\right) \pi ^{2}}{\mathfrak{p}_1+1}Q.
\end{equation}
for $c>0$. Plugging \eqref{Qinfty} into the expression $b$ defined in (\ref{b1}),we see
that $b<0$ as $Q\rightarrow \infty $. Also, it can be shown that the same result holds in the limit of small $\rho$.
Theorem \ref{Complex MHD} is proved. 
\section{Physical Remarks and Conclusions}
In this work, we investigate several transition scenarios of the magnetohydrodynamics (MHD) equations. As is well known, for (MHD) due to non-selfadjoint linear opeartor,  the transition can be caused by a finite set of real or nonreal eigenvalues crossing zero. 

When the first eigenvalue is real and simple, the transition depends on the character of the first critical index. In this case, the transition can only be Type-I or Type-II depending on a number exactly given by \eqref{a,b}. In particular, when the first critical eigenmode has a roll structure, the type of transition is independent of the Prandtl number $\mathfrak{p}_1$. In this case, the transition is Type-I if $Q>Q_{\ast}$ or $\mathfrak{p}_2>\mathfrak{p}_{\ast}$ where $Q_{\ast}$ and $\mathfrak{p}_{\ast}$ depends on the length scales of the domain. We find that $Q_{\ast}<307$ and $\mathfrak{p}_{\ast}<2.24$ regardless of the length scales of the domain. As $\max\{L_1,L_2\}\uparrow\infty$, $Q_{\ast}\downarrow4\pi^2$ and $\mathfrak{p}_{\ast}\downarrow2/\sqrt{3}$.

Next we study the case where there are two critical real eigenvalues. In this case we only consider the special geometry 
\[
\frac{L_1}{L_2}=\frac{j}{k \sqrt{3}}. \] with positive integers $j$, $k$ and $L_1$, $L_2$ denoting the horizontal length scales of the box.
With this assumption, it is possible that  two modes which can characterize a hexagon pattern become unstable at the same critical parameter. In this case we find that all types of transitions are possible in a total of eight different transition scenarios. However, in our numerical investigation, we encountered only two of these scenarios. We find that when $Q<Q_{\ast}$ and $\mathfrak{p}_2<\mathfrak{p}_{\ast}$ the system moves from a Type-III transition regime to a Type-I regime as $\mathfrak{p}_2$ crosses $\mathfrak{p}_{\ast}$. The minimal attractors in the stable domain of the Type-III transition regime have a steady rectangular pattern. In the Type-I transition regime, the rectangles and rolls are minimal attractors, and hexagons are unstable patterns after the transition. In this case we prove that for $\mathfrak{p}_2\geq 8$, the transition is always Type-I with rolls and rectangles as stable patterns and hexagons as unstable patterns after the transition. However $\mathfrak{p}_2\geq 8$ is a crude estimate and our numerical investigation suggests that this type of Type-I transition will be preferred for $\mathfrak{p}_2\geq\mathfrak{p}_2^{\ast}$ where $\mathfrak{p}_2^{\ast}<2.24$.

Finally, we consider the case where the first eigenvalue is simple and nonreal. This is always the case when $\mathfrak{p}_2<1$ and $Q>Q_0$. We only consider a roll type critical eigenmode. In this case the transition can be Type-I or Type-II. In particular for $Q$ sufficiently large or when the oscillation frequency $\rho$ is sufficiently small, the transition is Type-I and the transition structure is a time periodic roll pattern.

\bibliographystyle{amsplain}
\bibliography{mhd}

\providecommand{\bysame}{\leavevmode\hbox to3em{\hrulefill}\thinspace}
\providecommand{\MR}{\relax\ifhmode\unskip\space\fi MR }
\providecommand{\MRhref}[2]{%
  \href{http://www.ams.org/mathscinet-getitem?mr=#1}{#2}
}
\providecommand{\href}[2]{#2}
\begin{thebibliography}{10}

\bibitem{banerjee1985principle}
M.B. Banerjee, JR~Gupta, RG~Shandil, SK~Sood, B.~Banerjee, and K.~Banerjee,
  \emph{{On the principle of exchange of stabilities in the magnetohydrodynamic
  simple B{\'e}nard problem}}, Journal of mathematical analysis and
  applications \textbf{108} (1985), no.~1, 216--222.

\bibitem{chandrasekhar}
S.~Chandrasekhar, \emph{Hydrodynamic and hydromagnetic stability}, Dover
  Publications, Inc., 1981.

\bibitem{cross93}
M.C. Cross and P.C. Hohenberg, \emph{{Pattern formation outside of
  equilibrium}}, Reviews of Modern Physics \textbf{65} (1993), no.~3,
  851--1112.

\bibitem{dauby1993}
PC~Dauby, G.~Lebon, P.~Colinet, and J.C. Legros, \emph{Hexagonal marangoni
  convection in a rectangular box with slippery walls}, The Quarterly Journal
  of Mechanics and Applied Mathematics \textbf{46} (1993), no.~4, 683.

\bibitem{getling1998rayleigh}
A.V. Getling, \emph{{Rayleigh-B{\'e}nard convection: structures and dynamics}},
  World Scientific Pub Co Inc, 1998.

\bibitem{kosch}
EL~Koschmieder, \emph{{B{\'e}nard cells and Taylor vortices}}, Cambridge Univ
  Pr, 1993.

\bibitem{ptd}
Tian Ma and Shouhong Wang, \emph{Phase transition dynamics in nonlinear
  sciences}, submitted.

\bibitem{b-book}
\bysame, \emph{Bifurcation theory and applications}, World Scientific Series on
  Nonlinear Science. Series A: Monographs and Treatises, vol.~53, World
  Scientific Publishing Co. Pte. Ltd., Hackensack, NJ, 2005. \MR{MR2310258}

\bibitem{Proctor1982}
M.R.E. Proctor and N.O. Weiss, \emph{Magnetoconvection}, Reports on Progress in
  Physics \textbf{45} (1982), 1317--1379.

\bibitem{temam88}
Roger Temam, \emph{Infinite-dimensional dynamical systems in mechanics and
  physics}, second ed., Applied Mathematical Sciences, vol.~68,
  Springer-Verlag, New York, 1997. \MR{98b:58056}

\end{thebibliography}

\end{document}